\documentclass[12pt]{article}%
\usepackage{graphicx}
\usepackage{amsmath}
\usepackage{amsfonts}
\usepackage{amssymb}%
\newtheorem{theorem}{Theorem}

\newtheorem{lemma}[theorem]{Lemma}

\newtheorem{proposition}[theorem]{Proposition}

\newenvironment{proof}[1][Proof]{\noindent\textbf{#1.} }{\ \rule{0.5em}{0.5em}}
\begin{document}

\title{\textbf{Spontaneous} \textbf{Resonances and the Coherent States of the Queuing
Networks}}
\author{Alexander Rybko$^{1}$, Senya Shlosman$^{1,2}$, Alexander Vladimirov$^{1}$\\$^{1}$ Inst. of the Information Transmission Problems,\\Russian Academy of Sciences,\\Moscow, Russia\\$^{2}$ Centre de Physique Theorique, UMR 6207,\\CNRS, Luminy, Marseille, France}
\maketitle

\begin{abstract}
We present an example of a highly connected closed network of servers, where
the time correlations do not vanish in the infinite volume limit. The limiting
interacting particle system behaves in a periodic manner. This phenomenon is
similar to the continuous symmetry breaking at low temperatures in statistical
mechanics, with the average load playing the role of the inverse temperature.


MSC-class: 82C20 (Primary), 60J25 (Secondary)

\end{abstract}

\section{Introduction}

\subsection{Interacting particle systems with long range memory}

The theory of phase transitions, among many results, substantiates  the
possibility of constructing reliable systems from non-reliable elements. As an
example, consider the infinite volume stochastic Ising model at low
temperature $T$ in dimension $\geq2,$ see \cite{L}. It is well known, that if
we start this system from the configuration of all pluses, then the evolution
under Glauber dynamics has the property that the fraction of plus spins at any
time exceeds $\frac{1+m^{\ast}\left(  T\right)  }{2},$ which is bigger than
$\frac{1}{2}$ for $T<T_{cr.}$ (Here $m^{\ast}\left(  T\right)  $ is the
spontaneous magnetization.) On the other hand, if we consider finite volume
Ising model (with empty boundary condition, say), then this property does not
hold, and the system, started from the all plus state, will be found in the
state with the majority of the spins to be minuses at some later (random)
times. Therefore, the infinite system can remember, to some extent, its
initial state, while the finite system can not.

There are many other examples of this kind, that belong to the theory of
interacting particle systems, such as voter model, contact model, etc. In all
these examples we see systems, that are capable of \textquotedblleft
remembering\textquotedblright\ their initial state for arbitrary long times.

In the present paper we are constructing a particle system that
\textquotedblleft remembers its initial phase\textquotedblright. The rough
analogy can be described as follows. Imagine a brownian particle
$\varphi\left(  t\right)  ,$ with a unit drift, which lives on a circle.
Suppose the initial phase $\varphi\left(  0\right)  =0.$ Then the mean phase
$\bar{\varphi}\left(  t\right)  =t\mathrm{\operatorname{mod}}\left(
2\pi\right)  ,$ but with time we know the phase $\varphi\left(  t\right)  $
less and less precisely, since its variance grows, and in the limit
$t\rightarrow\infty$ the distribution of $\varphi\left(  t\right)  $ tends to
the uniform one. However, one can combine infinitely many such particles by
introducing suitable interaction between them in such a way that the memory of
the initial phase does not vanish and persists in time. Namely, one has to put
such particles at the sites of $\mathbb{Z}^{3}$ and to introduce the
attractive interaction between them. If the initial state is chosen to be
coherent, then the phase of every particle will grow linearly, while its
variance will stay bounded.

This is roughly what we do in the present paper. We consider a network of
simple servers which are processing messages. Since the service time of every
message is random, in the course of time each single server loses the memory
of its initial state. So, in particular, the network of non-interacting
servers, started in the same state, becomes de-synchronized after a finite
time. However, if one introduces certain natural interconnection between
servers, then it can happen that they are staying synchronized after an
arbitrary long time, thus breaking some generally believed properties of large
networks. We have to add here that such a phenomenon is possible only if the
mean number of particles per server is high enough; otherwise the infinite
network becomes de-synchronized, no matter which kind of interaction between
servers takes place. So the parameter of the mean number of particles per
server, called hereafter \textit{the load}, plays the same role as the
temperature in the statistical mechanics.

In other words, the transition we describe happens due to the fact that at low
load the behavior of our system is governed by the fixed point of the
underlying dynamical system, while at high load the dominant role is played by
its periodic attractor. A similar phenomenon was described by Hepp and Lieb in
\cite{HL}.

Below we present the simplest example of the above behavior. But we believe
that the phenomenon we describe is fairly general. Its origin lies in the fact
that any large network of the general type possesses some kind of the
continuous symmetry in the infinite limit, and it is  breaking of that
symmetry at high load that causes the long-range order behavior of the
network. In our case this is the rotation symmetry, corresponding to the
periodic orbit of the limiting dynamical system.

\subsection{Information networks and their collective behavior}

Now we will describe one pattern of behavior of certain large networks, which
was assumed to be universal. It is known under the name of Poisson Hypothesis.

The Poisson Hypothesis is a device to predict the behavior of large queuing
networks. It was formulated first by L. Kleinrock in \cite{K}, and concerns
the following situation. Suppose we have a large network of servers, through
which many customers are traveling, being served at different nodes of the
network. If the node is busy, the customers wait in the queue. Customers are
entering into the network from the outside via some nodes, and these external
flows of customers are Poissonian, with constant rates. The service time at
each node is random, depending on the node, and the customer. The PH
prediction about the (long-time, large-size) behavior of the network is the following:

\begin{itemize}
\item consider the total flow $\mathcal{F}$ of customers to a given node
$\mathcal{N}.$ Then $\mathcal{F}$ is approximately equal to a Poisson flow,
$\mathcal{P},$ with a time dependent rate function $\lambda_{\mathcal{N}%
}\left(  T\right)  .$

\item The exit flow from $\mathcal{N}$ -- \textbf{not} Poissonian in general!
-- has a rate function $\gamma_{\mathcal{N}}\left(  T\right)  ,$ which is
smoother than $\lambda_{\mathcal{N}}\left(  T\right)  $ (due to averaging,
taking place at the node $\mathcal{N}$).

\item As a result, the flows $\lambda_{\mathcal{N}}\left(  T\right)  $ at
various nodes $\mathcal{N}$ should go to a constant limits $\bar{\lambda
}_{\mathcal{N}}\approx\frac{1}{T}\int_{0}^{T}\lambda\left(  t\right)  dt$, as
$T\rightarrow\infty,$ the flows to different nodes being almost independent.

\item The above convergence is uniform in the size of the network.
\end{itemize}

Note that the distributions of the service times at the nodes of the network
can be arbitrary, so PH deals with quite a general situation. The range of
validity of PH is supposed to be the class of networks where the internal flow
to every node $\mathcal{N}$ is a union of flows from many other nodes, and
each one of these flows constitutes only a small fraction of the total flow to
$\mathcal{N}.$ If true, PH provides one with means to make easy computations
of quantities of importance in network design.

The rationale behind this conjectured behavior is natural: since the inflow is
a sum of many small inputs, it is approximately Poissonian. And due to the
randomness of the service time the outflow from each node should be
\textquotedblleft smoother\textquotedblright\ than the total inflow to this
node. (This statement was proven in \cite{RShV} under quite general
conditions.) In particular, the variation of the latter should be smaller than
that of the former, and so all the flows should converge to corresponding
constant values.

In the paper \cite{RSh} the Poisson Hypothesis is proven for simple networks
in the infinite volume limit, under some natural conditions. For systems with
constant service times it was proven earlier in \cite{St1}.

The purpose of the present paper is to construct a network that satisfies the
above assumption --  that the flow to every given node is an \textquotedblleft
infinite\textquotedblright\ sum of \textquotedblleft infinitesimally
small\textquotedblright\ flows from other nodes -- but has coherent states.
That means that the states of the servers are evolving in a synchronous
manner, and the \textquotedblleft phase\textquotedblright\ of a given server
behaves (in the thermodynamic limit -- i.e. in the limit of infinite network)
as a periodic non-random function, the same for different servers.

We have to stress that our network exhibits these coherent states only in the
regime when the average number $N$ of the customers per server -- called in
what follows \textit{the load -- }is large. For low load we expect the
convergence to the unique stationary state. This \textquotedblleft high
temperature\textquotedblright\ kind of behavior will be the subject of the
forthcoming work.

Our network $\nabla_{\infty}$ is constructed from infinitely many elementary
\textquotedblleft triangular\textquotedblright\ networks $\nabla$ (described
below, in Section \ref{nabla}). A single triangle network $\nabla=\nabla_{1}$
with $N$ customers is just a Markov continuous time ergodic jump process with
finitely many states. As $N$ becomes large, this Markov process tends (in the
appropriate \textquotedblleft Euler\textquotedblright\ limit) to a
(5-dimensional) dynamical system $\Delta$, possessing a periodic trajectory
$\mathcal{C},$ which turns out to be a stable (local) attractor. The
coordinate $\varphi$ parameterizing that attractor $\mathcal{C}$ is the
\textquotedblleft phase\textquotedblright\ alluded to in the previous
subsection. The combined network corresponds in the same sense to the coupled
family $\Delta_{\infty}$ of dynamical systems $\Delta.$ We establish the
synchronization property of that coupled family $\Delta_{\infty},$ and that
allows us to construct coherent states of the network $\nabla_{\infty}.$

The networks $\nabla_{M},$ composed of $M$\textbf{ }triangle networks
$\nabla,$ are\textbf{ }ergodic. Their evolution is given by irreducible finite
state Markov processes with continuous time. Let $\pi_{M}$ be the invariant
measure of the process $\nabla_{M}$. As $M\rightarrow\infty$, the sequence of
Markov processes $\nabla_{M}$ converges weakly on finite time intervals to a
certain limiting (non-linear Markov) process $\nabla_{\infty}$.\textbf{ }By
the theorem of Khasminsky -- see Theorem 1.2.14 in \cite{L} -- any
accumulation point of the sequence $\pi_{M}$ is a stationary measure of
$\nabla_{\infty}$. The special measure $\chi_{\infty}$, describing
\textquotedblleft the Poisson Hypothesis behavior\textquotedblright, is also a
stationary measure of $\nabla_{\infty}$. If $\chi_{\infty}$ is a global
attractor of $\nabla_{\infty}$, then, of course, the Poisson Hypothesis holds.
The proof of the Poisson Hypothesis in \cite{RSh} was based on this argument.
The existence of an accumulation point of the sequence $\pi_{M}$ that differs
from $\chi_{\infty}$ would be the strongest counterexample to the Poisson
Hypothesis. This problem will be addressed in forthcoming papers. Here we
prove a weaker statement that $\chi_{\infty}$ is not a global attractor for
$\nabla_{\infty}$.

\textbf{ }In \cite{RSt} Rybko and Stolyar observed that the condition that the
workload at every node of a multiclass open queueing network is less than $1$
is not sufficient for the network to be ergodic. In connection with this, they
introduced a new approach to the analysis of ergodicity of networks, which
reduces the problem to the question of stability of the associated fluid
models. It was shown by them that the two-node priority network, considered in
\cite{RSt}, is ergodic if and only if for every initial state of the
corresponding fluid model the total amount of fluid vanishes eventually. This
approach was further developed by Dai \cite{D}, Stolyar \cite{St2}, and
Puhalsky and Rybko \cite{PR}, who proved that stability of the fluid model is
necessary and sufficient for ergodicity of a certain class of general
networks. Interesting instances of non-ergodic queueing networks with mean
load being smaller than the capacity, where considered by Bramson \cite{B1,
B2}. Our construction will be based on the following open network introduced
by Rybko and Stolyar (RS-network) in \cite{RSt}.

This queuing network with four types of customers is represented by the
following 4-dimensional Markov process. Customers arrive to the network
according to Poisson inflows of constant rate $\lambda.$ The network consists
of two nodes -- $\bar{A}$ and $\bar{B}$. All the service times are
exponential, hence the network is defined by the rates, the evolution of types
of the customers and the priorities. The customer of type $A$ (respectively,
$B$) arrives to the node $\bar{A}$ (respectively, $\bar{B}$). The customer $A$
is served with the rate $\gamma_{A},$ then is sent to $\bar{B},$ with type
$AB.$ There he is served with the rate $\gamma_{AB}$ and leaves the network.
Symmetrically, $\gamma_{B}=\gamma_{A},$ and $\gamma_{BA}=\gamma_{AB}.$ Each
customer $AB$ is served before all the customers $B,$ while each customer $BA$
is served before all the customers $A$. The nominal workload at nodes $\bar
{A}$ and $\bar{B}$ equals $\rho=\lambda(\gamma_{A}^{-1}+\gamma_{BA}^{-1})$.
The service rates satisfy the conditions $\gamma_{AB}<2\lambda$ and $\rho<1$.
It is proved in \cite{RSt} that for certain values of the parameters the
resulting Markov process is transient. The fluid limit (or the Euler limit) of
this network evolves in the following non-trivial manner: each node is empty
during a positive fraction of time, but at other moments it is non-empty, and,
moreover, the total amount of the fluid in the network tends linearly to infinity.

The rest of the paper is organized as follows. In Section 2 we define our
networks $\nabla_{M}^{N}.$ Here $M$ is the size of the network and $N$ is the
load per node. We formulate the preliminary version of our Main Result. In
Section 3 we study the limiting network, $\nabla_{\infty}^{N},$ and prove the
convergence $\nabla_{M}^{N}\rightarrow\nabla_{\infty}^{N}.$ In Section 4 we
introduce the fluid networks, $\Delta_{M},$ which are coupled dynamical
systems, and their limit, $\Delta_{\infty},$ which turns out to be a
non-linear dynamical system, in the sense made precise in this Section. In
particular, we show that $\Delta_{\infty}$ is not ergodic. In the next Section
5 we prove the convergence of the Non-Linear Markov Process $\nabla_{\infty
}^{N}$ to its Euler fluid limit, $\Delta_{\infty},$ as $N\rightarrow\infty.$
The last Section 6 contains the formulation and the proof of our main result,
Theorem \ref{MR}.

To save on notation, we consider throughout this paper the simplest elementary
symmetric model, depending on 3 parameters. We stress the fact that this
(discrete) symmetry is not essential in our case, and our results are valid
for any small 6D-perturbation of our model.

\section{Mean-field network and its limit}

\subsection{Basic network \label{nabla}}

We will consider the following 5-dimensional Markov process, $\nabla^{N}$. It
describes a closed queuing network with $N$ customers. It consists of three
nodes: $\bar{O},$ $\bar{A}$ and $\bar{B},$ through which the customers go. All
the service times are exponential, so we only need to specify the rates, the
evolution of types of the customers and the priorities. To simplify the
presentation we will make a specific choice of these rates. The node $\bar{O}$
serves all the customers on the FIFO basis, with the rate $\gamma_{O}=3.$
After being served, the customer goes to the node $\bar{A}$ or to $\bar{B},$
choosing one of them with probability $\frac{1}{2}.$ If he arrives to $\bar
{A},$ he gets the type $A,$ otherwise $B.$ The customer $A$ is served with the
rate $\gamma_{A}=10,$ then is sent to $\bar{B},$ with type $AB.$ There he is
served with the rate $\gamma_{AB}=2$ and goes back to $\bar{O}.$
Symmetrically, $\gamma_{B}=10,$ and $\gamma_{BA}=2.$ Each customer $AB$ is
served before all the customers $B,$ and each customer $BA$ is served before
all the customers $A$. More precisely, if an $AB$ customer arrives to the
$\bar{B}$ node, while the node is serving some $B$ customer, his service is
stopped and is resumed only at the moment when the service of all $AB$
customers is over.%

%TCIMACRO{\FRAME{dpFUX}{1.7582in}{2.5036in}{0pt}{\Qcb{The elementary network.}%
%}{}{trinetrates.svg.eps}{\special{ language "Scientific Word";
%type "GRAPHIC";  maintain-aspect-ratio TRUE;  display "USEDEF";
%valid_file "F";  width 1.7582in;  height 2.5036in;  depth 0pt;
%original-width 1.721in;  original-height 2.463in;  cropleft "0";
%filename 'trinetRates.svg.eps';file-properties "XNPEU";}} }%
%BeginExpansion
\begin{center}
\fbox{\includegraphics[
height=2.5036in,
width=1.7582in
]%
{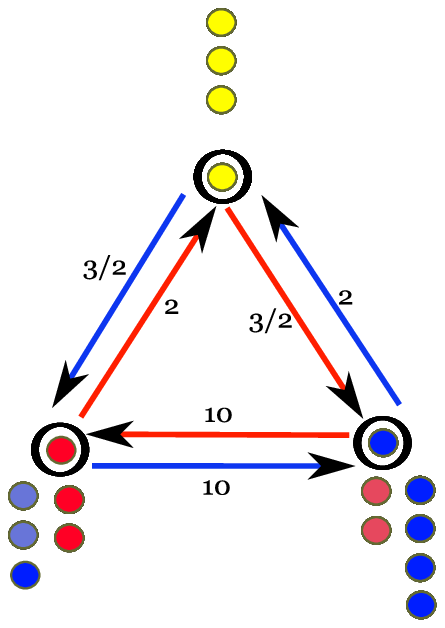}%
}\\
The elementary network.
\end{center}
%EndExpansion

Of course, the above choice of the rates is not the only possible. Any other
choice would be as good, provided the corresponding fluid network, which can
be associated to our queuing network, has some specific property -- namely, we
need this fluid network to have a \textbf{cyclic regime}. The fluid network
will be described in details in the Section \ref{S21} below.

\subsection{$M$ coupled processes}

Let $\nabla_{M}^{N}\ $be the Markov process, obtained from $M$ copies of
$\nabla^{N},$ interconnected in the mean-field manner. We take the total
number of customers to be $NM.$

The mean field network is defined as follows.\textbf{ } Each node $\bar{O}%
_{i},$ $i=1,...,M$, is connected to \textit{all }of the nodes $\bar{A}_{j},$
$\bar{B}_{j},$ $j=1,...,M$, and each customer, leaving the node $\bar{O}_{i},$
goes to each of the $2M$ nodes $\bar{A}_{j},$ $\bar{B}_{j}$ with the same
probability $\frac{1}{2M}.$ The rate of leaving the node $\bar{O}_{i}$ is the
same, as above, i.e. equals to $\gamma_{O}=3.$ In a similar way, the $A$
customers of every node $\bar{A}_{i}$ are exiting it with the rate
$\gamma_{AB}=10,$ and then choose one of the $\bar{B}_{j}$ nodes with
probability $\frac{1}{M},$ and so on. The priorities are kept the same: if the
node $\bar{A}_{i},$ say, is in the state with customers of both kinds -- $A$
and $BA$ -- present, then the $BA$ customers are served first, with no delay.%

%TCIMACRO{\FRAME{dtbpFUX}{1.7435in}{2.3644in}{0pt}{\Qcb{Two coupled
%processes.}}{}{duo_trinet_mf.svg.eps}{\special{ language "Scientific Word";
%type "GRAPHIC";  maintain-aspect-ratio TRUE;  display "USEDEF";
%valid_file "F";  width 1.7435in;  height 2.3644in;  depth 0pt;
%original-width 1.7054in;  original-height 2.3237in;  cropleft "0";
%filename 'duo_trinet_mf.svg.eps';file-properties "XNPEU";}} }%
%BeginExpansion
\begin{center}
\fbox{\includegraphics[
height=2.3644in,
width=1.7435in
]%
{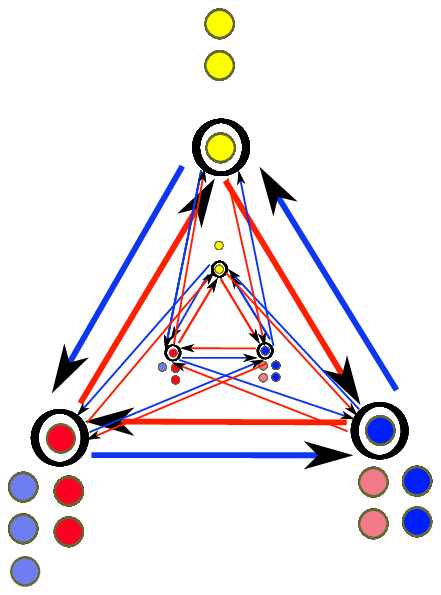}%
}\\
Two coupled processes.
\end{center}
%EndExpansion

The configuration of the process is given by the number of customers of each
type at each of the $3M$ nodes, that is by an integer point in $\left(
\mathbb{R}^{5M}\right)  ^{+}.$ Due to the mean-field symmetry, we can factor
the set of configurations by the product of permutation groups $S_{M}\times
S_{M}\times S_{M},$ and still have the Markov process. The orbit of the
symmetry group corresponds to a collection of $M^{3}$ integer points $\bar
{x}_{i}\in\left(  \mathbb{R}^{5}\right)  ^{+},$ some of which may coincide.

It is convenient for us to index these configurations by the atomic measures,%
\begin{equation}
\left\{  \bar{x}_{i}\right\}  \leadsto\frac{1}{M^{3}}\sum_{i=1}^{M^{3}}%
\delta_{\frac{\bar{x}_{i}}{N}}.\label{016}%
\end{equation}
In fact, they belong to the set $\mathcal{M}\left(  \left(  \frac{1}%
{N}\mathbb{Z}^{5}\right)  ^{+}\right)  .$ Note that every such measure $\mu$
factor into a product
\[
\mu\equiv\left(  \mu_{O},\mu_{\bar{A}},\mu_{\bar{B}}\right)  \equiv\mu
_{O}\times\mu_{\bar{A}}\times\mu_{\bar{B}}\equiv\Pi_{\bar{O}}\left[
\mu\right]  \times\Pi_{\bar{A}}\left[  \mu\right]  \times\Pi_{\bar{B}}\left[
\mu\right]
\]
of probability measures on $\mathbb{R}^{1}=\left\{  x_{O}\right\}  ,$ resp.
$\mathbb{R}^{2}=\left\{  x_{A},x_{BA}\right\}  $ and $\mathbb{R}^{2}=\left\{
x_{B},x_{AB}\right\}  .$ Here we denote by $\Pi_{\ast}$-s the various
projections (or marginals). We have $\mu_{O}=\frac{1}{M}\sum_{i=1}^{M}%
\delta_{\frac{\bar{x}_{i}}{N}}$ for some (not necessarily distinct) $\bar
{x}_{i}\in\mathbb{Z}^{1},$ $i=1,...,M,$ likewise $\mu_{\bar{A}}=\frac{1}%
{M}\sum_{i=1}^{M}\delta_{\frac{\bar{x}_{i}^{\prime}}{N}},$ $\mu_{\bar{B}%
}=\frac{1}{M}\sum_{i=1}^{M}\delta_{\frac{\bar{x}_{i}^{\prime\prime}}{N}},$
$\bar{x}_{i}^{\prime},\bar{x}_{i}^{\prime\prime}\in\mathbb{Z}^{2}.$ We will
denote the set of all such measures by $\mathcal{M}_{M}.$ The state $\nu
_{M}^{N}$ of our Markov process is then an element from $\mathcal{M}\left(
\mathcal{M}_{M}\right)  ,$ i.e. a measure on the measure space. Among these
there are configurations of the process $\nabla_{M},$ namely, the $\delta
$-measures $\delta_{m},$ with $m\in\mathcal{M}_{M},$ so we can define in this
way the embedding $\mathcal{M}_{M}\subset\mathcal{M}\left(  \mathcal{M}%
_{M}\right)  ,$ However, even if the initial state $\nu_{M}^{N}\left(
0\right)  $ of $\nabla_{M}$ happen to be such a measure $\delta_{m},$ i.e.
$\nu_{M}^{N}\left(  0\right)  \in\mathcal{M}_{M},$ then at any positive $t$ we
have only that $\nu_{M}^{N}\left(  t\right)  \in\mathcal{M}\left(
\mathcal{M}_{M}\right)  ,$ while in general $\nu_{M}^{N}\left(  t\right)
\notin\mathcal{M}_{M}.$

For the future use we will write down the rates of the factor-process. Let $v$
be some measure of the form $\left(  \ref{016}\right)  ,$ while $v^{\prime}$
be the measure obtained from $v$ after a single jump of the initial process.
For example, let us consider the case when the jump in question is of
$AB\rightarrow O$ type, from $\bar{B}$-type server to $\bar{O}$ server (with
the rate $\gamma_{AB}$). That means that for some unique well-defined (by the
pair $v,v^{\prime}$) elements $x_{\bar{B}}=\left(  x_{B},x_{AB}\right)
\in\left(  \frac{1}{N}\mathbb{Z}^{2}\right)  ^{+},$ $x_{O}\in$ $\left(
\frac{1}{N}\mathbb{Z}^{1}\right)  ^{+}$ we have:%
\begin{equation}
v^{\prime}\left(  x_{\bar{B}}\right)  =v\left(  x_{\bar{B}}\right)  -\frac
{1}{M},\ v^{\prime}\left(  x_{O}\right)  =v\left(  x_{O}\right)  -\frac{1}%
{M}.\label{061}%
\end{equation}
Of course, for another pair: $\tilde{x}_{\bar{B}}=\left(  x_{B},x_{AB}%
-1\right)  ,$ $\tilde{x}_{O}=x_{O}+1$, we have
\begin{equation}
v^{\prime}\left(  \tilde{x}_{\bar{B}}\right)  =v\left(  \tilde{x}_{\bar{B}%
}\right)  +\frac{1}{M},\ v^{\prime}\left(  \tilde{x}_{O}\right)  =v\left(
\tilde{x}_{O}\right)  +\frac{1}{M},\label{062}%
\end{equation}
while at all other locations the two measures are the same. Since there are
$Mv\left(  x_{\bar{B}}\right)  $ locations where the jump could originate, and
the fraction of sites with the desirable outcome is $v\left(  x_{O}\right)  ,$
we have for the rate $c\left(  v,v^{\prime}\right)  \equiv c_{AB}\left(
v,v^{\prime}\right)  $ the expression%
\[
c\left(  v,v^{\prime}\right)  =\gamma_{AB}Mv\left(  x_{\bar{B}}\right)
v\left(  x_{O}\right)  .
\]
If the measures $v,v^{\prime}$ are not related by $\left(  \ref{061}\right)
-\left(  \ref{062}\right)  ,$ then the rate $c_{AB}\left(  v,v^{\prime
}\right)  =0.$

We will keep the notation $\nabla_{M}$ for the factor-process.

\subsection{$M\rightarrow\infty$ limit: Non-linear Markov Process}

Suppose that a sequence of initial states $\nu_{M}\left(  0\right)
\in\mathcal{M}\left(  \mathcal{M}_{M}\right)  $ of the Markov processes
$\nabla_{M}$ is given, which satisfy $\nu_{M}\left(  0\right)  =\delta_{m_{M}%
},$ with $m_{M}\in\mathcal{M}_{M},$ and moreover the weak limit $\nu
=\lim_{M\rightarrow\infty}m_{M}$ exists. Then the weak limits $\nu\left(
t\right)  =\lim_{M\rightarrow\infty}\nu_{M}\left(  t\right)  $ exist for every
$t,$ and, moreover, for every $t$ we have $\nu\left(  t\right)  \in
\mathcal{M}.$ This is the Non-Linear Markov Process, NLMP, $\nabla_{\infty}.$
The process is called Non-Linear since the transition mechanism to evolve from
a given configuration depends not only on that configuration, but also on the
measure from which this configuration was drawn. Such processes were
introduced in \cite{M1, M2}, see also \cite{RSh}. The above limiting NLMP-s
depend on the parameter $N,$ which is the number of clients per basic queuing
network. We want to study the dependence on $N,$ so we explicitly
(re)introduce the index $N$ in our notation. Thus, $\nu^{N}\left(  t\right)  $
refers to the states of the process $\nabla_{\infty}^{N}.$

We will describe the limiting NLMP in the next Section \ref{TK}. Now we can
formulate the preliminary version of our main result.

\begin{theorem}
Consider the Non-Linear Markov Process $\nabla_{\infty},$ started from the
measure $\nu_{0}^{N},$ which is close enough to the atomic measure with the
single atom at vector\textbf{ }$\bar{X}\left(  A,N\right)  \in\left(  \frac
{1}{N}\mathbb{Z}^{5}\right)  ^{+},$\textbf{ }having coordinate\textbf{ }%
$x_{A}=1$ and all other coordinates zero. \textquotedblleft Close
enough\textquotedblright\ here means that for some $\varepsilon>0$ small
enough we have $\rho_{KROV}\left(  \nu_{0}^{N},\delta_{\bar{X}\left(
A,N\right)  }\right)  <\varepsilon.$ Suppose additionally that the $\alpha
$-exponential moment of the measure $\nu_{0}^{N}$ is less than a certain
quantity $E;$ $\alpha=\alpha\left(  \varepsilon\right)  ,$ $E=E\left(
\varepsilon\right)  .$ Then the measure $\nu_{t}^{N}$ does not converge to any
limit as $t\rightarrow\infty,$ provided $N$ is large enough.

More precisely, there exists a sequence of times $t_{k}^{\prime}%
\rightarrow\infty,$ such that
\[
\nu_{t_{k}^{\prime}}^{N}\left[  U_{N}\left(  \bar{X}\left(  A,N\right)
\right)  \right]  >1-\delta_{N},
\]
with $\delta_{N}\rightarrow0$ as $N\rightarrow\infty.$ Here $U_{N}\left(
\bar{X}\left(  A,N\right)  \right)  $ is a neighborhood of $\bar{X}\left(
A,N\right)  $ of radius $\varkappa_{N},$ with $\varkappa_{N}\rightarrow0$ as
$N\rightarrow\infty.$ At the same time, there exists another sequence
$t_{k}^{\prime\prime}\rightarrow\infty,$ for which $\nu_{t_{k}^{\prime\prime}%
}^{N}\left[  U_{N}\left(  \bar{X}\left(  B,N\right)  \right)  \right]
>1-\delta_{N}.$ In words, the measure $v_{t}$ exhibits oscillations.

Accordingly, the states of finite size networks, $\left(  \nu_{M}^{N}\right)
_{t},$ exhibit oscillations for long times, before going to their limits. The
duration of the oscillation regime diverges with $M.$ Different components of
$\left(  \nu_{M}^{N}\right)  _{t}$ are oscillating almost coherently, for
large $M.$
\end{theorem}

\section{The convergence$\ \nabla_{M}^{N}$ $\rightarrow\nabla_{\infty}^{N}:$
application of the Trotter-Kurtz theorem. \label{TK}}

Here we prove the convergence of the Markov processes $\nabla_{M}^{N}$ to the
Non-Linear Markov process $\nabla_{\infty}^{N}.$ We will do that by writing
down their generators $A_{M}$ and $A,$ and by subsequent application of the
Trotter-Kurtz theorem (Prop. 1.3.3 in \cite{EK}), which we formulate now.

Let $A_{M},A:X\rightarrow X$ are (unbounded) operators on the Banach space
$X,$ and $X_{0}\subset X$ is a dense subspace, belonging to the domains of
definition of all $A_{M}$-s and $A.$ The following two conditions are
sufficient for the convergence of the semigroups $\exp\left\{  tA_{M}\right\}
\rightarrow\exp\left\{  tA\right\}  $ on $X$ as $M\rightarrow\infty:$

\begin{enumerate}
\item $\forall\psi\in X_{0}$ we have $A_{M}\left(  \psi\right)  \rightarrow
A\left(  \psi\right)  $ as $M\rightarrow\infty;$

\item there exists a dense subspace $X_{1}\subset X_{0},$ such that
$\forall\psi\in X_{1}$ we have $\exp\left\{  tA\right\}  \left(  \psi\right)
\in X_{0}.$ Such subspace $X_{0}$ is called a \textit{core} of $A.$
\end{enumerate}

\subsection{Equation for the evolution $\nabla_{\infty}^{N}.$}

Here we study the limiting process $\nabla_{\infty}^{N}.$ We write down its
generator, and we exhibit its core.

Let $\nu_{t}$ be the evolution of the measure under $\nabla_{\infty}^{N}.$ To
find it we have to specify the initial measure $\nu_{0}$ and then to solve the
Cauchy problem for the differential equation, which equation we will write now.

To do it we first introduce the (Poisson) rates

\noindent$\bar{\lambda}\left(  t\right)  =\left(  \lambda_{O}\left(  t\right)
,\lambda_{A}\left(  t\right)  ,\lambda_{B}\left(  t\right)  ,\lambda
_{AB}\left(  t\right)  ,\lambda_{BA}\left(  t\right)  \right)  ,$
corresponding to the state $\nu_{t}:$\textbf{ }%
\begin{equation}
\lambda_{a}\left(  t\right)  =\gamma_{a}\sum_{x:x_{a}>0}\nu_{t}\left(
x\right)  ,\text{ for }a=O,AB,BA, \label{012}%
\end{equation}%
\begin{equation}
\lambda_{A}\left(  t\right)  =\gamma_{A}\sum_{x:x_{A}>0,x_{BA}=0}\nu
_{t}\left(  x\right)  ,\ \ \lambda_{B}\left(  t\right)  =\gamma_{B}%
\sum_{x:x_{B}>0,x_{AB}=0}\nu_{t}\left(  x\right)  . \label{012a}%
\end{equation}
We also introduce the 5D vectors $\Delta_{a}$ to be the basis vectors of the
lattice $\mathbb{Z}^{5}.$ Then%
\begin{align}
\frac{d\nu_{t}\left(  x\right)  }{dt}  &  =-\nu_{t}\left(  x\right)  \left(
\sum_{a=O,A,B,AB,BA}\lambda_{a}\left(  t\right)  \right) \nonumber\\
&  -\nu_{t}\left(  x\right)  \left(  \sum_{a=O,AB,BA}\gamma_{a}\left(
1-\delta_{x_{a}}\right)  +\gamma_{A}\left(  1-\delta_{x_{A}}\right)
\delta_{x_{BA}}+\gamma_{B}\left(  1-\delta_{x_{B}}\right)  \delta_{x_{AB}%
}\right) \nonumber\\
&  +\nu_{t}\left(  x-\Delta_{O}\right)  \left(  1-\delta_{x_{O}}\right)
\left(  \lambda_{AB}\left(  t\right)  +\lambda_{BA}\left(  t\right)  \right)
\label{0111}\\
&  +\sum_{a=A,B}\nu_{t}\left(  x-\Delta_{a}\right)  \left(  1-\delta_{x_{a}%
}\right)  \frac{\lambda_{O}\left(  t\right)  }{2}\nonumber\\
&  +\nu_{t}\left(  x-\Delta_{AB}\right)  \left(  1-\delta_{x_{AB}}\right)
\lambda_{A}\left(  t\right)  +\nu_{t}\left(  x-\Delta_{BA}\right)  \left(
1-\delta_{x_{BA}}\right)  \lambda_{B}\left(  t\right) \nonumber\\
&  +\sum_{a=O,AB,BA}\nu_{t}\left(  x+\Delta_{a}\right)  \mathbf{\gamma}%
_{a}\nonumber\\
&  +\nu_{t}\left(  x+\Delta_{A}\right)  \mathbf{\gamma}_{A}\delta_{x_{BA}}%
+\nu_{t}\left(  x+\Delta_{B}\right)  \mathbf{\gamma}_{B}\delta_{x_{AB}%
}.\nonumber
\end{align}

This is the value of the function $A\varphi_{x},$ where $A$ is the generator
of the Markov semigroup $S_{t}=\exp\left\{  tA\right\}  $ of the process
$\nabla_{\infty}^{N},$ and the function $\varphi_{x}$, which on every measure
$\nu\in\mathcal{M}\left(  \mathbb{Z}^{5+}\right)  $ takes value $\nu\left(
x\right)  ,$ computed at the point $\nu_{t}.$ As we will show below, the
system $\left(  \ref{0111}\right)  $ has a unique solution.

For reasons which will be explained later, it will be more convenient for us
to use another basis in the space of functions on measures. Namely, for every
$\nu\in\mathcal{M}\left(  \mathbb{Z}^{5+}\right)  $ and every $x\in
\mathbb{Z}^{5+}$ we define the function $u\left(  x\right)  =\sum_{y\geq x}%
\nu\left(  y\right)  ,$ where the summation goes over all sites $y$ such that
all the coordinates of the difference $y-x$ are non-negative. Then the
functions $\lambda_{a}\left(  t\right)  $ (see $\left(  \ref{012}\right)  $)
are given in the new variables as
\[
\lambda_{a}\left(  t\right)  =\gamma_{a}u_{t}\left(  \Delta_{a}\right)  ,
\]
while the action of the generator $A$ on the function $u$ is given by the
following (simpler) equation:
\begin{align}
\frac{du_{t}\left(  x\right)  }{dt}  &  =-\sum_{a=O,AB,BA}\left(  u_{t}\left(
x\right)  -u_{t}\left(  x+\Delta_{a}\right)  \right)  \gamma_{a}\left(
1-\delta_{x_{a}}\right) \nonumber\\
&  -\left(  u_{t}\left(  x\right)  -u_{t}\left(  x+\Delta_{A}\right)
-u_{t}\left(  x+\Delta_{AB}\right)  +u_{t}\left(  x+\Delta_{A}+\Delta
_{AB}\right)  \right)  \gamma_{A}\left(  1-\delta_{x_{A}}\right)
\delta_{x_{BA}}\nonumber\\
&  -\left(  u_{t}\left(  x\right)  -u_{t}\left(  x+\Delta_{B}\right)
-u_{t}\left(  x+\Delta_{BA}\right)  +u_{t}\left(  x+\Delta_{B}+\Delta
_{BA}\right)  \right)  \gamma_{B}\left(  1-\delta_{x_{B}}\right)
\delta_{x_{AB}}\nonumber\\
&  +\left(  u_{t}\left(  x-\Delta_{O}\right)  -u_{t}\left(  x\right)  \right)
\left(  1-\delta_{x_{O}}\right)  \left(  \gamma_{AB}u_{t}\left(  \Delta
_{AB}\right)  +\gamma_{BA}u_{t}\left(  \Delta_{BA}\right)  \right)
\label{014}\\
&  +\sum_{a=A,B}\left(  u_{t}\left(  x-\Delta_{a}\right)  -u_{t}\left(
x\right)  \right)  \left(  1-\delta_{x_{a}}\right)  \frac{\gamma_{O}%
u_{t}\left(  \Delta_{O}\right)  }{2}\nonumber\\
&  +\left(  u_{t}\left(  x-\Delta_{AB}\right)  -u_{t}\left(  x\right)
\right)  \left(  1-\delta_{x_{AB}}\right)  \gamma_{A}u_{t}\left(  \Delta
_{A}\mathbf{+}\Delta_{BA}\right) \nonumber\\
&  +\left(  u_{t}\left(  x-\Delta_{BA}\right)  -u_{t}\left(  x\right)
\right)  \left(  1-\delta_{x_{BA}}\right)  \gamma_{B}u_{t}\left(  \Delta
_{B}\mathbf{+}\Delta_{AB}\right)  .\nonumber
\end{align}
The first three lines correspond to the second line of the equation $\left(
\ref{0111}\right)  ,$ while the last four -- to the lines 3--5; the remaining
lines of it disappear from the equations for $u.$ The advantage of $\left(
\ref{014}\right)  $ over $\left(  \ref{0111}\right)  $ is that the equation
for $\frac{du_{t}\left(  x\right)  }{dt}$ contains only $u_{t}\left(
y\right)  $-s with $y$-s in some finite set $Y\left(  x\right)  ,$ and
moreover $\max_{x}\left\vert Y\left(  x\right)  \right\vert =20.$

Of course, the coordinates $u\left(  \cdot\right)  $-s on $\mathcal{M}\left(
\mathbb{Z}^{5+}\right)  $ are not independent. There are two kinds of
relations between them:

\begin{enumerate}
\item every value $\nu\left(  x\right)  $ equals to $L_{x}\left(  u\right)  ,$
where $L$ is a certain linear form, depending on $u\left(  y\right)  $ with
$y$-s having form $y=x+\sum e_{a}\Delta_{a},$ $e_{a}=0,1;$ we need that for
all $x$ $\ L_{x}\left(  u\right)  \geq0;$

\item
\begin{equation}
\lim_{x\rightarrow\infty}u\left(  x\right)  =0. \label{015}%
\end{equation}

\end{enumerate}

For technical reasons we will extend the action of our Markov semigroup to the
space $\mathcal{M}\left(  \mathbb{K}\right)  $ of measures on the
compactification $\mathbb{K}$ of the lattice $\mathbb{Z}^{5+},$ where
\[
\mathbb{K}=\left\{  \mathbb{Z}^{+}+\infty\right\}  ^{5}.
\]
The functions $\left\{  u\left(  x\right)  ,x\in\mathbb{Z}^{5+}\right\}  $ on
$\mathcal{M}\left(  \mathbb{K}\right)  $ also play the role of coordinates
there, provided that the relation $\left(  \ref{015}\right)  $ is dropped. The
evolution of the measures is given by the same set of equations $\left(
\ref{014}\right)  .$

We supply $\mathcal{M}\left(  \mathbb{K}\right)  $ with the topology of weak
convergence. (We repeat for clarity that the subset $\mathcal{M}\left(
\mathbb{Z}^{5+}\right)  \subset\mathcal{M}\left(  \mathbb{K}\right)  $ is
invariant under our semigroup.) Let $\mathcal{C}^{0}=\mathcal{C}\left(
\mathcal{M}\left(  \mathbb{K}\right)  \right)  $ be the space of functions on
$\mathcal{M}\left(  \mathbb{K}\right)  ,$ continuous with respect to this topology.

\begin{theorem}
The semigroup $S_{t}$ acts on the space $\mathcal{C}=\mathcal{C}\left(
\mathcal{M}\left(  \mathbb{K}\right)  \right)  $ of continuous functions on
$\mathcal{M}\left(  \mathbb{K}\right)  ,$ and is strongly continuous and contracting.
\end{theorem}

\begin{proof}
\textbf{1. }Let us show the existence of the solutions to $\left(
\ref{014}\right)  .$ The equations $\left(  \ref{014}\right)  $ are describing
the evolution of the \textquotedblleft closed\textquotedblright\ system.
Consider now the corresponding \textquotedblleft open\textquotedblright%
\ system, defined by the (arbitrary) rates $\bar{\lambda}=\lambda_{a}\left(
t\right)  $ of the Poisson inflows and the initial state $u_{0}.$ It evolves
according to the equations
\begin{align*}
\frac{du_{t}\left(  x\right)  }{dt}  &  =-\sum_{a=O,AB,BA}\left(  u_{t}\left(
x\right)  -u_{t}\left(  x+\Delta_{a}\right)  \right)  \gamma_{a}\left(
1-\delta_{x_{a}}\right) \\
&  -\left(  u_{t}\left(  x\right)  -u_{t}\left(  x+\Delta_{A}\right)
-u_{t}\left(  x+\Delta_{AB}\right)  +u_{t}\left(  x+\Delta_{A}+\Delta
_{AB}\right)  \right)  \gamma_{A}\left(  1-\delta_{x_{A}}\right)
\delta_{x_{BA}}\\
&  -\left(  u_{t}\left(  x\right)  -u_{t}\left(  x+\Delta_{B}\right)
-u_{t}\left(  x+\Delta_{BA}\right)  +u_{t}\left(  x+\Delta_{B}+\Delta
_{BA}\right)  \right)  \gamma_{B}\left(  1-\delta_{x_{B}}\right)
\delta_{x_{AB}}\\
&  +\left(  u_{t}\left(  x-\Delta_{O}\right)  -u_{t}\left(  x\right)  \right)
\left(  1-\delta_{x_{O}}\right)  \left(  \lambda_{AB}\left(  t\right)
+\lambda_{BA}\left(  t\right)  \right) \\
&  +\sum_{a=A,B}\left(  u_{t}\left(  x-\Delta_{a}\right)  -u_{t}\left(
x\right)  \right)  \left(  1-\delta_{x_{a}}\right)  \frac{\lambda_{O}\left(
t\right)  }{2}\\
&  +\left(  u_{t}\left(  x-\Delta_{AB}\right)  -u_{t}\left(  x\right)
\right)  \left(  1-\delta_{x_{AB}}\right)  \lambda_{A}\left(  t\right) \\
&  +\left(  u_{t}\left(  x-\Delta_{BA}\right)  -u_{t}\left(  x\right)
\right)  \left(  1-\delta_{x_{BA}}\right)  \lambda_{B}\left(  t\right)  .
\end{align*}
\textbf{ }The corresponding exit rates $b_{a}$ are given by the natural
relations%
\[
b_{a}^{\bar{\lambda}}\left(  t\right)  =\gamma_{a}u_{t}\left(  \Delta
_{a}\right)  .
\]
Consider the function $\bar{d}^{\bar{\lambda}}\left(  t\right)  :$%
\[
d_{O}^{\bar{\lambda}}\left(  t\right)  =b_{AB}^{\bar{\lambda}}\left(
t\right)  +b_{BA}^{\bar{\lambda}}\left(  t\right)  ,
\]%
\[
d_{A}^{\bar{\lambda}}\left(  t\right)  =d_{B}^{\bar{\lambda}}\left(  t\right)
=\frac{1}{2}b_{O}^{\bar{\lambda}}\left(  t\right)  ,
\]%
\[
d_{AB}^{\bar{\lambda}}\left(  t\right)  =b_{A}^{\bar{\lambda}}\left(
t\right)  ,~d_{BA}^{\bar{\lambda}}\left(  t\right)  =b_{B}^{\bar{\lambda}%
}\left(  t\right)  .
\]
The closed system is a fixed point of the map $\bar{\lambda}\overset
{\psi_{u_{0}}}{\leadsto}\bar{d}^{\bar{\lambda}},$ i.e. a solution of the
equation%
\[
\bar{\lambda}=\bar{d}^{\bar{\lambda}}.
\]
To see the existence of a fixed point, let us introduce the functions
$\bar{\Lambda},\bar{B}^{\bar{\Lambda}}$ and $\bar{D}^{\bar{\Lambda}}:$
\[
\Lambda_{a}\left(  t\right)  =\int_{0}^{t}\lambda_{a}\left(  t\right)
dt,~B_{a}^{\bar{\Lambda}}\left(  t\right)  =\int_{0}^{t}b_{a}^{\bar{\lambda}%
}\left(  t\right)  \left(  t\right)  dt,~D_{a}^{\bar{\Lambda}}\left(
t\right)  =\int_{0}^{t}d_{a}^{\bar{\lambda}}\left(  t\right)  dt,
\]
and the corresponding mapping $\bar{\Lambda}\overset{\Psi_{u_{0}}}{\leadsto
}\bar{D}^{\bar{\Lambda}}.$ The functions $\bar{\Lambda},\bar{B}^{\bar{\Lambda
}}$ and $\bar{D}^{\bar{\Lambda}}$ are monotone continuous, moreover, the
functions $\bar{B}^{\bar{\Lambda}}$ and $\bar{D}^{\bar{\Lambda}}$ have
uniformly bounded derivatives. Let $\mathfrak{c}$ be the upper bound for these
derivatives, and $\mathfrak{C}$ be the space of all continuous monotone 5D
vector-functions on $\left[  0,T\right]  ,$ vanishing at zero, with the
derivatives bounded by $\mathfrak{c}$ once they exist. Then $\mathfrak{C}$ is
compact and convex, therefore the map $\Psi_{u_{0}}:\mathfrak{C}%
\rightarrow\mathfrak{C}$ has at least one fixed point.

\textbf{2. }We now will show that for every $u_{0}$ the map $\Psi_{u_{0}}$ is
a contraction; that will imply the uniqueness of the solution. Without loss of
generality we can assume that $T$ is small. Informally, the contraction takes
place because the exit rates $b_{a}^{\bar{\lambda}}\left(  t\right)  $ for
$t\in\left[  0,T\right]  $ with $T$ small depend mainly on the initial state
$u_{0}:$ the new clients, arriving during the time $\left[  0,T\right]  $ have
no chance to be served before $T,$ if there were clients already waiting.
Therefore the \textquotedblleft worst\textquotedblright\ case for us is when
the initial state $\nu_{0}$ is the measure $\delta_{0},$ having a unit atom at
$0\in\mathbb{Z}^{5+}.$

So let $\lambda_{1}\left(  t\right)  ,$ $\lambda_{2}\left(  t\right)
,t\in\left[  0,T\right]  $ be the rates of two Poisson inflows to the empty
server, and $\gamma$ be the service rate. We want to estimate the difference
$b_{1}\left(  t\right)  -b_{2}\left(  t\right)  $ of the rates of the exit
flows. We can couple the two service processes in the following way: let
$\lambda\left(  t\right)  =\min$ $\left\{  \lambda_{1}\left(  t\right)
,\lambda_{2}\left(  t\right)  \right\}  .$ Then we write $\lambda_{i}\left(
t\right)  =\lambda\left(  t\right)  +\eta_{i}\left(  t\right)  ,$ where
$\eta_{i}\left(  t\right)  =\lambda_{i}\left(  t\right)  -\lambda\left(
t\right)  .$ We will call the clients arriving with the rate $\lambda\left(
t\right)  $ as colorless, and we call the $\eta_{1}\left(  t\right)  $ clients
as red, while the $\eta_{2}\left(  t\right)  $ clients as blue. The colorless
clients have priority in their service: if a colorless client arrives, then
all the colored ones have to wait -- even the one currently under the service.
Then the difference $\left\vert b_{1}\left(  t\right)  -b_{2}\left(  t\right)
\right\vert $ is bounded from above by the sum of the exit rates of colored
clients, which does not exceed
\[
\left\vert b_{1}\left(  t\right)  -b_{2}\left(  t\right)  \right\vert
\leq\gamma\mathbf{\Pr}\left(
\begin{array}
[c]{c}%
\text{server is occupied by a }\\
\text{colored client at the moment }t
\end{array}
\right)  \leq\gamma\int_{0}^{t}\left\vert \lambda_{1}\left(  t\right)
-\lambda_{2}\left(  t\right)  \right\vert dt.
\]
Hence we have contraction with the contraction rate at most $\gamma T,$ which
is small for small $T.$ We denote by $\bar{\lambda}_{u}\left(  t\right)  $ the
unique fixed point of $\Psi_{u}.$

\textbf{3.} Finally we prove that the semigroup preserves the space of
continuous functions. First of all we observe that the map $\Psi_{u}$ depends
on the initial measure $u$ in a continuous way. Therefore the same is true for
$\bar{\lambda}_{u}\left(  t\right)  ,$ the fixed point of $\Psi_{u}.$ Hence
$u\left(  t\right)  $ depends continuously on $u\left(  0\right)  .$
\end{proof}

Let us consider the subspace $\mathcal{C}^{2}\subset\mathcal{C}^{0}$ of
functions $f$, which have the following properties:

\begin{enumerate}
\item for every $x\in\mathbb{Z}^{5+}$ the function $f$ has the first
derivative $\frac{\partial f}{\partial u\left(  x\right)  }$;

\item for every $x,y\in\mathbb{Z}^{5+}$ the function $f$ has the second
derivative $\frac{\partial^{2}f}{\partial u\left(  x\right)  \partial u\left(
y\right)  }$;

\item all these derivatives are bounded, uniformly in $x,y.$
\end{enumerate}

It is easy to see that the set of coordinate functions $\left\{  u\left(
x\right)  ,x\in\mathbb{Z}^{5+}\right\}  $ on $\mathcal{M}\left(
\mathbb{K}\right)  $ can distinguish any two measures from $\mathcal{M}\left(
\mathbb{K}\right)  $. Due to the compactness of $\mathcal{M}\left(
\mathbb{K}\right)  $ we can apply the Stone-Weierstrass theorem, which implies
that the subspace $\mathcal{C}^{2}\subset\mathcal{C}^{0}$ is dense in
$\mathcal{C}^{0}.$ We now will show the following

\begin{proposition}
\label{cor} For every $t$ we have $S_{t}\left(  \mathcal{C}_{0}^{2}\right)
\subset\mathcal{C}^{2},$ where the subspace $\mathcal{C}_{0}^{2}%
\subset\mathcal{C}^{2}$ consists of all functions depending only on finitely
many variables $\left\{  u\left(  x\right)  \right\}  .$ In particular, the
subspace $\mathcal{C}^{2}$ is a core of the generator $A.$\medskip\ 
\end{proposition}

\textbf{Proof. }To do this we will use the Proposition 1 of the paper
\cite{DKV}:

\begin{lemma}
\label{karp} Consider the infinite system of equations%
\[
\frac{d}{dt}z_{k}\left(  t\right)  =\sum_{i}a_{ki}\left(  t\right)
z_{i}\left(  t\right)  +b_{k}\left(  t\right)  ,\ t\geq0.
\]
Suppose that for all $k$%
\[
\sum_{i}\left\vert a_{ki}\left(  t\right)  \right\vert \leq a,\ \left\vert
b_{k}\left(  t\right)  \right\vert \leq b_{0}\exp\left\{  bt\right\}
,\ \left\vert z_{k}\left(  0\right)  \right\vert \leq c,
\]
with $a<b.$ Then%
\[
\left\vert z_{k}\left(  t\right)  \right\vert \leq c\exp\left\{  at\right\}
+\frac{b_{0}}{b-a}\left(  \exp\left\{  bt\right\}  -\exp\left\{  at\right\}
\right)  .
\]
\medskip
\end{lemma}

From $\left(  \ref{014}\right)  $ it follows immediately, that
\[
\frac{d}{dt}\left(  \frac{\partial u_{t}\left(  v\right)  }{\partial
u_{0}\left(  x\right)  }\right)  =\sum_{w\in Y\left(  v\right)  }\bar{a}%
_{vw}\left(  t\right)  \left(  \frac{\partial u_{t}\left(  w\right)
}{\partial u_{0}\left(  x\right)  }\right)  ,
\]
with \textbf{ }$\sum_{w\in Y\left(  v\right)  }\left\vert \bar{a}_{vw}\left(
t\right)  \right\vert <\hat{a},$ $\left\vert \frac{\partial u_{0}\left(
v\right)  }{\partial u_{0}\left(  x\right)  }\right\vert \leq1,$ for some
$\hat{a}<\infty,$ uniformly in $v,$ so Lemma \ref{karp} applies, and all the
derivatives $\left\vert \frac{\partial u_{t}\left(  v\right)  }{\partial
u_{0}\left(  x\right)  }\right\vert $ are uniformly bounded, provided $t<T.$
Further on,%
\[
\frac{d}{dt}\left(  \frac{\partial^{2}u_{t}\left(  v\right)  }{\partial
u_{0}\left(  x\right)  \partial u_{0}\left(  y\right)  }\right)  =\sum_{w\in
Y\left(  v\right)  }\bar{a}_{vw}\left(  t\right)  \left(  \frac{\partial
^{2}u_{t}\left(  w\right)  }{\partial u_{0}\left(  x\right)  \partial
u_{0}\left(  y\right)  }\right)  +\tilde{b}_{v}\left(  t\right)  ,
\]
with the same $\bar{a}_{vw}\left(  t\right)  $-s$,$ while the term $\tilde
{b}_{v}\left(  t\right)  ,$ consisting of the products of the first
derivatives $\frac{\partial u_{t}\left(  w^{\prime}\right)  }{\partial
u_{0}\left(  x\right)  }\frac{\partial u_{t}\left(  w^{\prime\prime}\right)
}{\partial u_{0}\left(  y\right)  },$ is also uniformly bounded, as was just
shown, provided $t<T.$ Therefore the derivatives $\left\vert \frac
{\partial^{2}u_{t}\left(  v\right)  }{\partial u_{0}\left(  x\right)  \partial
u_{0}\left(  y\right)  }\right\vert $ are uniformly bounded as well.

For other functions we just use the chain rule. $\blacksquare$

\subsection{Equation for the evolution $\nabla_{M}^{N}$ and the convergence}

Now we will write the generator $A_{M}$ of the process $\nabla_{M}^{N}.$ Let
$\psi\left(  \cdot\right)  $ be a function on $\mathcal{M}_{M}$. (In fact, we
need it to be defined on a smaller set $\mathcal{M}_{M}\cap\mathcal{M}\left(
\left(  \frac{1}{N}\mathbb{Z}^{5}\right)  ^{+}\right)  .$ Throughout this
section the value of $N$ will be fixed, and we will keep it just $1,$ in order
to simplify the notation.) We will introduce the following notations for the
increments of the measure $v=\left(  v_{\bar{O}},v_{\bar{A}},v_{\bar{B}%
}\right)  $ (which are themselves (signed) measures on $\mathbb{Z}^{1}$ or
$\mathbb{Z}^{2}$):
\[
\Delta_{O,x}\left(  y\right)  =\left\{
\begin{array}
[c]{cc}%
1 & \text{ if }y=x\\
-1 & \text{if }y=x-1\\
0 & \text{ otherwise}%
\end{array}
\right.  ,\ x,y\in\mathbb{Z}^{1},
\]%
\[
\Delta_{A,x_{\bar{A}}}\left(  y\right)  =\left\{
\begin{array}
[c]{cc}%
1 & \text{ if }y=x_{\bar{A}}\equiv\left(  x_{A},x_{BA}\right) \\
-1 & \text{if }y=\left(  x_{A}-1,x_{BA}\right) \\
0 & \text{ otherwise}%
\end{array}
\right.  ,\ x_{\bar{A}},y\in\mathbb{Z}^{2},
\]%
\[
\Delta_{BA,x_{\bar{A}}}\left(  y\right)  =\left\{
\begin{array}
[c]{cc}%
1 & \text{ if }y=x_{\bar{A}}\equiv\left(  x_{A},x_{BA}\right) \\
-1 & \text{if }y=\left(  x_{A},x_{BA}-1\right) \\
0 & \text{ otherwise}%
\end{array}
\right.  ,\ x_{\bar{A}},y\in\mathbb{Z}^{2},
\]
and similar definitions for the remaining measures $\Delta_{B,x_{\bar{B}}}$
and $\Delta_{AB,x_{\bar{B}}}.$ Then%

\begin{align}
&  \left(  A_{M}\psi\right)  \left(  v\right)  =\sum_{v^{\prime}}c\left(
v,v^{\prime}\right)  \left[  \psi\left(  v^{\prime}\right)  -\psi\left(
v\right)  \right] \label{601}\\
&  =\sum_{x_{O}\geq1}M\tfrac{\gamma_{O}}{2}v_{\bar{O}}\left(  x_{O}\right)
\sum_{x_{\bar{A}}}v_{\bar{A}}\left(  x_{\bar{A}}\right)  \left[  \psi\left(
v_{\bar{O}}-\tfrac{\Delta_{O,x_{O}}}{M},v_{\bar{A}}+\tfrac{\Delta
_{A,x_{\bar{A}}+\left(  1,0\right)  }}{M},v_{\bar{B}}\right)  -\psi\left(
v_{\bar{O}},v_{\bar{A}},v_{\bar{B}}\right)  \right] \nonumber\\
&  +\sum_{x_{O}\geq1}M\tfrac{\gamma_{O}}{2}v_{\bar{O}}\left(  x_{O}\right)
\sum_{x_{\bar{B}}}v_{\bar{B}}\left(  x_{\bar{B}}\right)  \left[  \psi\left(
v_{\bar{O}}-\tfrac{\Delta_{O,x_{O}}}{M},v_{\bar{A}},v_{\bar{B}}+\tfrac
{\Delta_{B,x_{\bar{B}}+\left(  1,0\right)  }}{M}\right)  -\psi\left(
v_{\bar{O}},v_{\bar{A}},v_{\bar{B}}\right)  \right] \nonumber\\
&  +\sum_{x_{\bar{A}}:x_{BA}\geq1}M\gamma_{BA}v_{\bar{A}}\left(  x_{\bar{A}%
}\right)  \sum_{x_{O}}v_{\bar{O}}\left(  x_{O}\right)  \left[  \psi\left(
v_{\bar{O}}+\tfrac{\Delta_{O,x_{O}+1}}{M},v_{\bar{A}}-\tfrac{\Delta
_{BA,x_{\bar{A}}}}{M},v_{\bar{B}}\right)  -\psi\left(  v_{\bar{O}},v_{\bar{A}%
},v_{\bar{B}}\right)  \right] \nonumber\\
&  +\sum_{x_{\bar{B}}:x_{AB}\geq1}M\gamma_{AB}v_{\bar{B}}\left(  x_{\bar{B}%
}\right)  \sum_{x_{O}}v_{\bar{O}}\left(  x_{O}\right)  \left[  \psi\left(
v_{\bar{O}}+\tfrac{\Delta_{O,x_{O}+1}}{M},v_{\bar{A}},v_{\bar{B}}%
-\tfrac{\Delta_{AB,x_{\bar{B}}}}{M}\right)  -\psi\left(  v_{\bar{O}}%
,v_{\bar{A}},v_{\bar{B}}\right)  \right] \nonumber\\
&  +\sum_{x_{A}\geq1}M\gamma_{A}v_{\bar{A}}\left(  x_{A},0\right)
\sum_{x_{\bar{B}}}v_{\bar{B}}\left(  x_{\bar{B}}\right)  \left[  \psi\left(
v_{\bar{O}},v_{\bar{A}}-\tfrac{\Delta_{A,\left(  x_{A},0\right)  }}{M}%
,v_{\bar{B}}+\tfrac{\Delta_{AB,x_{\bar{B}}+\left(  0,1\right)  }}{M}\right)
-\psi\left(  v_{\bar{O}},v_{\bar{A}},v_{\bar{B}}\right)  \right] \nonumber\\
&  +\sum_{x_{B}\geq1}M\gamma_{B}v_{\bar{B}}\left(  x_{B},0\right)
\sum_{x_{\bar{A}}}v_{\bar{A}}\left(  x_{\bar{A}}\right)  \left[  \psi\left(
v_{\bar{O}},v_{\bar{A}}+\tfrac{\Delta_{BA,x_{\bar{A}}+\left(  0,1\right)  }%
}{M},v_{\bar{B}}-\tfrac{\Delta_{B,\left(  x_{B},0\right)  }}{M}\right)
-\psi\left(  v_{\bar{O}},v_{\bar{A}},v_{\bar{B}}\right)  \right]  .\nonumber
\end{align}
Suppose now that the function $\psi$ is differentiable in each of the
variables $v_{\bar{O}}\left(  x_{O}\right)  ,x_{O}\in\mathbb{Z}^{1},v_{\bar
{A}}\left(  x_{\bar{A}}\right)  ,x_{\bar{A}}\in\mathbb{Z}^{2},v_{\bar{B}%
}\left(  x_{\bar{B}}\right)  ,x_{\bar{B}}\in\mathbb{Z}^{2}.$ Then each of the
six increments in $\left(  \ref{601}\right)  $ equals to the corresponding
derivative $\psi^{\prime}$ of $\psi,$ computed at some intermediate point. If
moreover the function $\psi$ is continuously differentiable (which is implied
by the twice differentiability), we can take a limit as $M\rightarrow\infty,$
obtaining the convergence to the limiting operator
\begin{align*}
&  \left(  A\psi\right)  \left(  v\right) \\
&  =\sum_{x_{O}\geq1,x_{\bar{A}}}\tfrac{\gamma_{O}}{2}v_{\bar{O}}\left(
x_{O}\right)  v_{\bar{A}}\left(  x_{\bar{A}}\right)  \left[  \tfrac
{\partial\psi\left(  v\right)  }{\partial\left(  v_{\bar{O}}\left(
x_{O}-1\right)  \right)  }-\tfrac{\partial\psi\left(  v\right)  }%
{\partial\left(  v_{\bar{O}}\left(  x_{O}\right)  \right)  }+\tfrac
{\partial\psi\left(  v\right)  }{\partial\left(  v_{\bar{A}}\left(
x_{A}+1,x_{BA}\right)  \right)  }-\tfrac{\partial\psi\left(  v\right)
}{\partial\left(  v_{\bar{A}}\left(  x_{A},x_{BA}\right)  \right)  }\right] \\
&  +\sum_{x_{O}\geq1,x_{\bar{B}}}\tfrac{\gamma_{O}}{2}v_{\bar{O}}\left(
x_{O}\right)  v_{\bar{B}}\left(  x_{\bar{B}}\right)  \left[  \tfrac
{\partial\psi\left(  v\right)  }{\partial\left(  v_{\bar{O}}\left(
x_{O}-1\right)  \right)  }-\tfrac{\partial\psi\left(  v\right)  }%
{\partial\left(  v_{\bar{O}}\left(  x_{O}\right)  \right)  }+\tfrac
{\partial\psi\left(  v\right)  }{\partial\left(  v_{\bar{B}}\left(
x_{B}+1,x_{AB}\right)  \right)  }-\tfrac{\partial\psi\left(  v\right)
}{\partial\left(  v_{\bar{B}}\left(  x_{B},x_{AB}\right)  \right)  }\right] \\
&  +\sum_{x_{\bar{A}}:x_{BA}\geq1,x_{O}}\gamma_{BA}v_{\bar{A}}\left(
x_{\bar{A}}\right)  v_{\bar{O}}\left(  x_{O}\right)  \left[  \tfrac
{\partial\psi\left(  v\right)  }{\partial\left(  v_{\bar{O}}\left(
x_{O}+1\right)  \right)  }-\tfrac{\partial\psi\left(  v\right)  }%
{\partial\left(  v_{\bar{O}}\left(  x_{O}\right)  \right)  }+\tfrac
{\partial\psi\left(  v\right)  }{\partial\left(  v_{\bar{A}}\left(
x_{A},x_{BA}-1\right)  \right)  }-\tfrac{\partial\psi\left(  v\right)
}{\partial\left(  v_{\bar{A}}\left(  x_{A},x_{BA}\right)  \right)  }\right] \\
&  +\sum_{x_{\bar{B}}:x_{AB}\geq1,x_{O}}\gamma_{AB}v_{\bar{B}}\left(
x_{\bar{B}}\right)  v_{\bar{O}}\left(  x_{O}\right)  \left[  \tfrac
{\partial\psi\left(  v\right)  }{\partial\left(  v_{\bar{O}}\left(
x_{O}+1\right)  \right)  }-\tfrac{\partial\psi\left(  v\right)  }%
{\partial\left(  v_{\bar{O}}\left(  x_{O}\right)  \right)  }+\tfrac
{\partial\psi\left(  v\right)  }{\partial\left(  v_{\bar{B}}\left(
x_{B},x_{AB}-1\right)  \right)  }-\tfrac{\partial\psi\left(  v\right)
}{\partial\left(  v_{\bar{B}}\left(  x_{B},x_{AB}\right)  \right)  }\right] \\
&  +\sum_{x_{\bar{B}},x_{\bar{A}}:x_{A}\geq1}\gamma_{A}v_{\bar{A}}\left(
x_{A},0\right)  v_{\bar{B}}\left(  x_{\bar{B}}\right)  \left[  \tfrac
{\partial\psi\left(  v\right)  }{\partial\left(  v_{\bar{A}}\left(
x_{A}-1,0\right)  \right)  }-\tfrac{\partial\psi\left(  v\right)  }%
{\partial\left(  v_{\bar{A}}\left(  x_{A},0\right)  \right)  }+\tfrac
{\partial\psi\left(  v\right)  }{\partial\left(  v_{\bar{B}}\left(
x_{B},x_{AB}+1\right)  \right)  }-\tfrac{\partial\psi\left(  v\right)
}{\partial\left(  v_{\bar{B}}\left(  x_{B},x_{AB}\right)  \right)  }\right] \\
&  +\sum_{x_{\bar{A}},x_{\bar{B}}:x_{B}\geq1,}\gamma_{B}v_{\bar{B}}\left(
x_{B},0\right)  v_{\bar{A}}\left(  x_{\bar{A}}\right)  \left[  \tfrac
{\partial\psi\left(  v\right)  }{\partial\left(  v_{\bar{B}}\left(
x_{B}-1,0\right)  \right)  }-\tfrac{\partial\psi\left(  v\right)  }%
{\partial\left(  v_{\bar{B}}\left(  x_{B},0\right)  \right)  }+\tfrac
{\partial\psi\left(  v\right)  }{\partial\left(  v_{\bar{A}}\left(
x_{A},x_{BA}+1\right)  \right)  }-\tfrac{\partial\psi\left(  v\right)
}{\partial\left(  v_{\bar{A}}\left(  x_{A},x_{BA}\right)  \right)  }\right]  .
\end{align*}
Let us apply this formula to the "coordinate" function $\psi\left(
\cdot\right)  =\phi_{y}\left(  \cdot\right)  ,$ where $\phi_{y}\left(
v\right)  =v\left(  y\right)  \equiv v_{O}\left(  y_{O}\right)  v_{\bar{A}%
}\left(  y_{\bar{A}}\right)  v_{B}\left(  y_{\bar{B}}\right)  .$ The result is
the one given by $\left(  \ref{0111}\right)  :$%
\begin{align*}
\left(  A\phi_{x}\right)  \left(  v\right)   &  =-v\left(  x\right)  \left(
\sum_{a=O,A,B,AB,BA}\lambda_{a}\left(  t\right)  \right) \\
&  -v\left(  x\right)  \left(  \sum_{a=O,AB,BA}\gamma_{a}\left(
1-\delta_{x_{a}}\right)  +\gamma_{A}\left(  1-\delta_{x_{A}}\right)
\delta_{x_{BA}}+\gamma_{B}\left(  1-\delta_{x_{B}}\right)  \delta_{x_{AB}%
}\right) \\
&  +v\left(  x-\Delta_{O}\right)  \left(  1-\delta_{x_{O}}\right)  \left(
\lambda_{AB}\left(  t\right)  +\lambda_{BA}\left(  t\right)  \right) \\
&  +\sum_{a=A,B}v\left(  x-\Delta_{a}\right)  \left(  1-\delta_{x_{a}}\right)
\frac{\lambda_{O}\left(  t\right)  }{2}\\
&  +v\left(  x-\Delta_{AB}\right)  \left(  1-\delta_{x_{AB}}\right)
\lambda_{A}\left(  t\right)  +v\left(  x-\Delta_{BA}\right)  \left(
1-\delta_{x_{BA}}\right)  \lambda_{B}\left(  t\right) \\
&  +\sum_{a=O,AB,BA}v\left(  x+\Delta_{a}\right)  \mathbf{\gamma}_{a}\\
&  +v\left(  x+\Delta_{A}\right)  \mathbf{\gamma}_{A}\delta_{x_{BA}}+v\left(
x+\Delta_{B}\right)  \mathbf{\gamma}_{B}\delta_{x_{AB}}.
\end{align*}
Since the space of functions of finitely many $v$-s coincides with that
depending of finitely many $u$-s, that finishes the proof.

\section{Fluid networks}

One of the key ingredients of the proof of out Main result is the
investigation of the fluid (Euler) limits of various networks. They are
introduced in the present Section.

\subsection{Fluid systems with one fluid}

The fluid systems with one fluid are the following simple dynamical systems.
Consider the containers $V_{1},V_{2},...,V_{n},$ filled (partially) with
water. Suppose that some pairs of these containers are connected by (directed)
pipes, through which the water can flow. On every pipe $i,j$ there is a pump
$\rho_{ij}$ working, with the capacity of sending $\gamma_{ij}\geq0$ units of
water per unit time from $V_{i}$ to $V_{j}$. The pumps are working constantly,
and if the container $V_{i}$ has less water than the pump $\rho_{ij}$ can
handle, the result is that the pump sucks in whatever there is. For example,
if the network is given by the graph
\[
V_{1}\overset{\gamma_{12}}{\rightarrow}V_{2}\overset{\gamma_{23}}{\rightarrow
}V_{3},
\]
with $\gamma_{12}=\frac{1}{2},$ $\gamma_{23}=1,$ then in a while the container
$V_{2}$ will be empty, and the flow along the pipe $23$ will be $\frac{1}{2}$
(provided the water supply in $V_{1}$ lasts). If a container $V_{i}$ has
several pipes $ij_{k}$ attached, then the water is shared by the pumps
$\rho_{ij_{k}}$ proportionally to the capacities $\gamma_{ij_{k}}$ (this is
relevant only in the situation when the level of water in $V_{i}$ is zero, and
all the incoming water immediately leaves it).

Suppose now that at the moment $T=0$ all the containers are filled with water
in the amounts of $v_{i}\left(  0\right)  ,$ and then we turn on all the
pumps. We are looking on the levels $v_{i}\left(  t\right)  ,$ as they are
changing in time. It turns out that there exists a time $T^{\prime}$
(depending on the system), after which the levels $v_{i}\left(  t\right)  $
become stable and do not change anymore. (Some of them can in fact be zero.)
In particular, such a system can never exhibit a cyclic behavior. Of course,
the stability of the levels does not imply that the water in the network does
not flow. It just means that for every container the amount of fluid entering
is equal to the amount of fluid leaving.

We will not provide the proof of this known statement (see \cite{D} and the
references there), since we will not use this fact. Below we will consider
fluid networks which do exhibit cyclic behavior, and we find such examples
among the fluid networks with several kinds of fluids.

\subsection{Basic fluid network \label{S21}}

We will consider the following 5-dimensional dynamical system, $\Delta$, which
is a closed version of the open RS-network. It consists of three nodes:
$\bar{O},$ $\bar{A}$ and $\bar{B},$ through which various fluids are passing.
The node $\bar{O}$ has one type of the fluid, in the amount $x_{O}\geq0,$ and
that fluid flows into the nodes $\bar{A}$ and $\bar{B}$ in equal amounts. The
rate of this flow $\gamma_{O}=3$, which means that three units of the fluid
$x_{O}$ leave $\bar{O}$ per unit time (provided the supply lasts, of course),
so each of the two nodes $\bar{A}$ and $\bar{B}$ gets $\frac{3}{2}$ units of
incoming fluids, $A$ and $B,$ per unit time. The amounts of these fluids are
denoted by $x_{A}$ and $x_{B}.$ The fluid $A$ then goes to the node $\bar{B},$
where it turns into the fluid $AB,$ while the fluid $B$ goes to $\bar{A}$ and
turns there into $BA.$ The corresponding rates are $\gamma_{A}=\gamma_{B}=10.$
The fluids $AB$ and $BA$ then go back to $\bar{O},$ with the rates
$\gamma_{AB}=\gamma_{BA}=2.$ The last thing which has to be specified is the
following priority: if the node $\bar{A}$ is in the state with both amounts
$x_{A}$ and $x_{BA}$ positive, then \textit{the fluid }$BA$\textit{ goes
first}. One can think that the fluid $BA$ is heavier and of higher viscosity
than $A,$ so it goes to the bottom of the node $\bar{A}$ and flows out first
(and relatively slow). The same applies to the node $\bar{B}.$%

%TCIMACRO{\FRAME{dtbpFUX}{2.9464in}{2.2857in}{0pt}{\Qcb{Basic fluid network.}%
%}{}{threefluids.svg.eps}{\special{ language "Scientific Word";
%type "GRAPHIC";  maintain-aspect-ratio TRUE;  display "USEDEF";
%valid_file "F";  width 2.9464in;  height 2.2857in;  depth 0pt;
%original-width 2.9023in;  original-height 2.2459in;  cropleft "0";
%filename 'threefluids.svg.eps';file-properties "XNPEU";}} }%
%BeginExpansion
\begin{center}
\fbox{\includegraphics[
height=2.2857in,
width=2.9464in
]%
{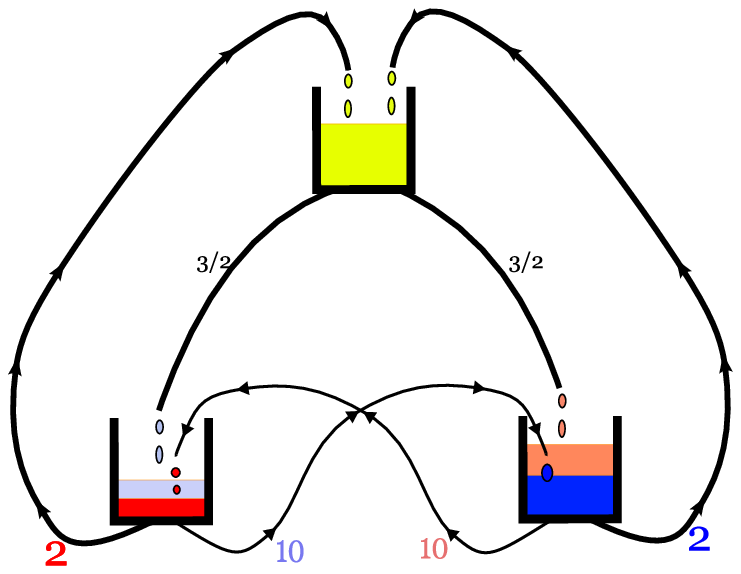}%
}\\
Basic fluid network.
\end{center}
%EndExpansion

As stated, the system is not well-defined. For example, the dynamics is not
specified if we try to start it from the configuration
\begin{equation}
x_{A}=a>0,x_{B}=b>0,x_{BA}=x_{AB}=0,x_{O}=1-a-b. \label{051}%
\end{equation}
The reason is that if the fluid $A$ \textquotedblleft starts
first\textquotedblright, then it will create some amount of the heavy fluid
$AB$ in $\bar{B},$ so the fluid $B$ in $\bar{B}$ will be blocked. The same
holds for the fluid $B$ \textquotedblleft starting first\textquotedblright%
.\textbf{ }

The precise definition of the system, given below, follows \cite{RSt, St2,
D}\textbf{. }Consider first the simplest case of a single node with capacity
$\gamma$ per unit time. Let the initial amount of fluid $x(0)$ be given, an
let $Y(t)$ be the net amount of fluid that arrives to the node during the time
interval $[0,t]$. In what follows we will call it the \textit{net inflow. }The
function $Y(t)$ is monotone non-decreasing function, $Y(0)=0,$ and we assume
that $Y(t)$ is Lipschitz continuous, $t\geq0$. It is easy to check that the
evolution of the amount $x(t)$ is given by
\begin{equation}
x(t)=V(t)+U(t), \label{EE6}%
\end{equation}
where $V(t)=x(0)+Y(t)-\gamma t$ is called the \textit{virtual level }of the
fluid, while $U(t)=\max\{0,-\inf_{s\in\lbrack0,t]}V(s)\}$ is the
\textit{unused service capacity }of our node.

We introduce for the correspondence $\left(  \ref{EE6}\right)  $ the notation
\begin{equation}
x(\cdot)=W(\gamma,x(0),Y(\cdot)), \label{EE6a}%
\end{equation}
that is, $W$ maps the initial fluid level and the inflow function to the
evolution of fluid level. As we will show later (in Lemma \ref{L2}),\textbf{
}the map $W$ is Lipschitz continuous map from $R\times C[0,\infty)$ to
$C[0,\infty)$.\textbf{ }

Further on, let $Z\left(  t\right)  $ be the total amount of fluid that leaves
the node during the time interval $[0,t]$: $Z(t)=x(0)+Y(t)-x(t)$. Again, the
function $Z\left(  t\right)  $ -- the \textit{net outflow -- } is monotone
non-decreasing Lipschitz continuous, with $Z(0)=0.$ By assumption, the
derivatives $z(t)=\dot{Z}(t),$ $y\left(  t\right)  =\dot{Y}(t),$ existing
a.e., satisfy
\[
z(t)=\left\{
\begin{array}
[c]{cc}%
\gamma & \text{ if }x(t)>0,\\
y\left(  t\right)  & \text{ otherwise.}%
\end{array}
\right.
\]
This property is the reason to call our discipline \textit{work-conserving};
the server is always working at its full capacity.

In the same way we can treat the node through which several fluids are
passing. Thus we introduce the vector $\bar{Y}(t)=\{Y_{1}(t),\dots,Y_{n}(t)\}$
of the net inflows during the time interval $[0,t]$, each $Y_{i}(t)$ being
non-decreasing and Lipschitz continuous. The vector $\bar{Z}(t)$ will denote
the corresponding collection of the net outflows. (Of course, it does depend
on the priorities of the fluids.) Consider again the derivatives
$y_{i}(t)=\dot{Y}_{i}(t)$ and $z_{i}(t)=\dot{Z}_{i}(t)$ (they exist for almost
all $t$ and define $Y_{i}(t)$ and $Z_{i}(t)$ in a unique way once we fix
$Y_{i}(0)$ and $Z_{i}(0)$ to be zero). Introduce also the workload rate by
$v(t)=\sum y_{i}(t)\gamma_{i}^{-1}$, where $\gamma_{i}$ are the service rates.
The service discipline of our node is work-conserving, if the following
property holds: once $\Vert x(t)\Vert>0$, we have
\begin{equation}
\sum z_{i}(t)\gamma_{i}^{-1}=1; \label{EE4}%
\end{equation}
otherwise
\begin{equation}
z_{i}(t)=y_{i}(t) \label{EE5}%
\end{equation}
for all $i$. The following statement is immediate:

\begin{proposition}
\label{PP1} Let $\bar{x}(0)=0$ and $v(t)\le1$ for almost all $t\ge0$. Then
$\bar{x}(t)=0$, $t\ge0$.
\end{proposition}

\textbf{ }The system described in the beginning of the present subsection
corresponds to the following specification of the above general formulation.
We have $\bar{Y}(t)\in\left(  \mathbb{R}^{5}\right)  ^{+},$\textbf{ }and in
terms of the map $W$ (see $\left(  \ref{EE6a}\right)  $) the evolution is
given by the equations:%
\begin{equation}
x_{O}(\cdot)=W(3,x_{O}(0),Y_{O}(\cdot)), \label{EE1}%
\end{equation}%
\begin{equation}
x_{BA}(\cdot)=W(2,x_{BA}(0),Y_{BA}(\cdot)), \label{EE2}%
\end{equation}%
\begin{equation}
x_{A}(\cdot)=W(10,5x_{BA}(0)+x_{A}(0),5Y_{BA}(\cdot)+Y_{A}(\cdot
))-5x_{BA}(\cdot), \label{EE3}%
\end{equation}
the equations for $x_{B}$ and $x_{AB}$ being symmetric. Then for the
derivatives, whenever they exist, we derive from (\ref{EE1})-(\ref{EE3}) that%

\[
\frac{d}{dt}x_{O}\left(  t\right)  =\left\{
\begin{array}
[c]{cc}%
y_{O}\left(  t\right)  -3 & \text{ if }x_{O}\left(  t\right)  >0,\text{ or if
}x_{O}\left(  t\right)  =0\text{ and }y_{O}\left(  t\right)  -3>0,\\
0 & \text{if }x_{O}\left(  t\right)  =0\text{ and }y_{O}\left(  t\right)
-3\leq0;
\end{array}
\right.
\]%
\[
\frac{d}{dt}x_{AB}\left(  t\right)  =\left\{
\begin{array}
[c]{cc}%
y_{AB}\left(  t\right)  -2 & \text{ if }x_{AB}\left(  t\right)  >0,\text{ or
if }x_{AB}\left(  t\right)  =0\text{ and }y_{AB}\left(  t\right)  -2>0,\\
0 & \text{if }x_{AB}\left(  t\right)  =0\text{ and }y_{AB}\left(  t\right)
-2\leq0;
\end{array}
\right.
\]%
\[
\frac{d}{dt}x_{A}\left(  t\right)  =\left\{
\begin{array}
[c]{cc}%
y_{A}\left(  t\right)  & \text{if }x_{BA}\left(  t\right)  >0,\\
& \\
y_{A}\left(  t\right)  -10\frac{2-y_{AB}\left(  t\right)  }{2} &
\begin{array}
[c]{c}%
\text{if }x_{A}\left(  t\right)  >0,~x_{BA}\left(  t\right)  =0,\text{ }%
y_{AB}\left(  t\right)  -2\leq0;\text{ or }\\
\text{if }x_{A}\left(  t\right)  =x_{BA}\left(  t\right)  =0,~y_{AB}\left(
t\right)  -2\leq0\text{ }\\
\text{and }y_{A}\left(  t\right)  -10\frac{2-y_{AB}\left(  t\right)  }{2}>0
\end{array}
\\
& \\
0 &
\begin{array}
[c]{c}%
\text{if }x_{A}\left(  t\right)  =x_{BA}\left(  t\right)  =0,~y_{AB}\left(
t\right)  -2\leq0\text{ }\\
\text{and }y_{A}\left(  t\right)  -10\frac{2-y_{AB}\left(  t\right)  }{2}%
\leq0;
\end{array}
\end{array}
\right.
\]
the equations for $x_{B}$ and $x_{BA}$ being symmetric.

In fact, for most points in the phase space the above differential equations
are sufficient to describe our dynamical system, since the set of values $t$
where one of the derivatives does not exist is nowhere dense. However, this is
not true for all points, and so we need to use more complicated equations
$\left(  \ref{EE1}-\ref{EE3}\right)  .$

The above relations define what we will call Non-Homogeneous Dynamical System,
NHDS. We will denote this dynamical system by $\Delta\left(  \bar{Y}\right)
,$ since it is driven by the inflow $\bar{Y}.$ This is just a usual
non-autonomous dynamical system. All the non-linear dynamical systems, which
will appear below, correspond to different choices of the flows $\bar{Y}.$

Let now $\bar{x}\left(  t\right)  $ be the trajectory of the NHDS,
corresponding to the initial state $\bar{x}\left(  0\right)  $ and the given
inflows $\bar{Y}\left(  t\right)  .$ We define the closed fluid network
evolution $\Delta$ of the point $\bar{x}\left(  0\right)  $ as the evolution
$\bar{x}\left(  t\right)  $ under any dynamics $\Delta\left(  \bar{Y}\left(
t\right)  \right)  ,$ for which the following relations between the inflows
$\bar{Y}\left(  t\right)  $ and the outflows
\begin{equation}
\bar{Z}(t)\equiv\bar{Z}^{\bar{x},\bar{Y}}(t)=\bar{x}(0)+\bar{Y}(t)-\bar{x}(t)
\label{z}%
\end{equation}
hold:%
\begin{equation}
Y_{O}\left(  t\right)  =Z_{AB}^{\bar{x},\bar{Y}}\left(  t\right)
+Z_{BA}^{\bar{x},\bar{Y}}\left(  t\right)  , \label{020}%
\end{equation}%
\begin{equation}
Y_{A}\left(  t\right)  =\frac{1}{2}Z_{O}^{\bar{x},\bar{Y}}\left(  t\right)  ,
\label{021}%
\end{equation}%
\begin{equation}
Y_{AB}\left(  t\right)  =Z_{A}^{\bar{x},\bar{Y}}\left(  t\right)  ,
\label{022}%
\end{equation}
and symmetric relations for $A$ and $BA$ variables. The set of all solutions
$\bar{Y}\left(  t\right)  $ of the equations $\left(  \ref{020}\right)
$--$\left(  \ref{022}\right)  $ will be denoted by $\mathcal{Y}\left(  \bar
{x}\left(  0\right)  \right)  .$

It is well known that the set $\mathcal{Y}\left(  \bar{x}\left(  0\right)
\right)  $ is not empty for any initial state $\bar{x}(0)$. We reproduce here
the proof, since analogous argument will be used throughout the paper.

\begin{proposition}
For any point $\bar{x}(0)$ there exists at least one trajectory of closed
fluid network evolution $\Delta,$ passing through it.
\end{proposition}

\begin{proof}
It suffices to prove that a solution exists in any given bounded time interval
$[0,T]$. For a given $\bar{x}=\bar{x}(0)$, consider the map $G:\bar{Y}%
(\cdot)\rightarrow\bar{Z}^{\bar{x},\bar{Y}}(\cdot).$ It is clearly a
continuous map of $C[0,T]$ into itself. The outflow $\bar{Z}(\cdot)$ is
Lipschitz continuous (by Lemma \ref{L2} below), with Lipschitz constant $L$
independent of $\bar{x}(0)$ or $\bar{Y}(\cdot).$ Hence $G$ takes the convex
compact set of $L$-Lipschitz continuous functions $\bar{Y}(\cdot)$ on $\left[
0,T\right]  $, satisfying $\bar{Y}(0)=0$, into itself. Therefore, by the
Brouwer theorem, the map $G$ has at least one fixed point.
\end{proof}

The system $\Delta$ has the following properties:

\begin{itemize}
\item The total amount of fluid $\left\vert \bar{x}\left(  t\right)
\right\vert =x_{O}\left(  t\right)  +x_{A}\left(  t\right)  +...+x_{B}\left(
t\right)  $ is evidently conserved. We will assume $\left\vert \bar{x}\left(
t\right)  \right\vert =1.$

\item For some initial states $\bar{x}\left(  0\right)  $ the equations
$\left(  \ref{020}\right)  $--$\left(  \ref{022}\right)  $ have multiple
solutions. This is true for all initial conditions $\bar{x}\left(  0\right)
,$ given by $\left(  \ref{051}\right)  .$ The equations $\left(
\ref{020}\right)  -\left(  \ref{022}\right)  $ have in this case three
solutions, so there are at least three dynamical systems, defined by $\bar
{x}\left(  0\right)  .$ The first solution corresponds to the flow rates
$y_{AB}=10,$ $y_{BA}=0,$ the other one is symmetric to the first one:
$y_{BA}=10,$ $y_{AB}=0,$ while the third one is given by the flow rates
$AB\rightarrow O,BA\rightarrow O,A\rightarrow AB,B\rightarrow BA$ all equal to
$\frac{10}{6},$ $O\rightarrow A,$ $O\rightarrow B$ equal to $\frac{3}{2}.$ One
can say, having that property in mind, that at some points the uniqueness of
the trajectory breaks down, so it can happen that for two trajectories
$\bar{x}^{\prime}\left(  t\right)  ,$ $\bar{x}^{\prime\prime}\left(  t\right)
$ we have $\bar{x}^{\prime}\left(  t\right)  =\bar{x}^{\prime\prime}\left(
t\right)  $ for $t\leq t_{0},$ but $\bar{x}^{\prime}\left(  t\right)  \neq
\bar{x}^{\prime\prime}\left(  t\right)  $ for $t>t_{0}.$ That just means that
there are two different (non-autonomous) dynamical systems, which have
trajectories, coinciding for $t\leq t_{0},$ but not for $t>t_{0}.$

\item The curve $\bar{x}\left(  t\right)  \equiv\left\{  1,0,0,0,0\right\}  $
is a trajectory, i.e. the point $\ast=\left\{  1,0,0,0,0\right\}  $ is a fixed
point of $\Delta.$ Its flow rates are constant; their values are: $y_{A}%
=y_{B}=y_{AB}=y_{BA}=\frac{3}{2},$ $y_{O}=3.$ Let us check that with these
flow rates the amount of the fluids in $\bar{A}$ and $\bar{B}$ will stay zero.
Indeed, $\frac{y_{A}}{\gamma_{A}}+\frac{y_{BA}}{\gamma_{BA}}=\frac{3/2}%
{10}+\frac{3/2}{2}=\frac{9}{10}<1.$ So our claim follows from the Proposition
\ref{PP1}. In fact, the nodes $\bar{A}$ and $\bar{B}$ are even underloaded,
which means that in the long run each node gets on the average less fluid than
its serving capacity is. Note also that there are trajectories $\bar{x}\left(
t\right)  $ such that\textbf{ }$\bar{x}\left(  t\right)  =\ast$ for $t\leq
t_{0}$, but $\bar{x}\left(  t\right)  \neq\ast$ for $t>t_{0},$ with any
$t_{0}.$

\item There are (not necessarily uniqueness) points in $\left(  \mathbb{R}%
^{5}\right)  ^{+},$ from where (some) trajectory goes to $\ast.$ One such
family is the set of points $U=\left\{  \bar{x}:0<x_{A}=x_{B}<\frac{1}%
{2},x_{AB}=\text{ }x_{BA}=0,x_{O}=1-x_{A}-x_{B}\right\}  ,$ and the flows are:
$AB\rightarrow O,BA\rightarrow O,A\rightarrow AB,B\rightarrow BA$ all equal to
$\frac{10}{6},$ $O\rightarrow A,$ $O\rightarrow B$ equal to $\frac{3}{2}.$ In
fact, from every point in $U$ two more trajectories start. The values of the
flow rates for one of them for small initial segment of time is given by:
$y_{A}=y_{B}=\frac{3}{2},$ $y_{AB}=10,$ $y_{BA}=0,$ $y_{O}=2.$ It is
describing the situation when the light fluid $A$ flows to the node $\bar{B},$
and the resulting heavy fluid $AB$ in the node $\bar{B}$ blocks the fluid $B$
in the node $\bar{B}$ from exiting. The second solution is obtained by
interchanging $A$ and $B.$\textbf{ }There is a bigger set $\bar{U}\supset U,$
having dimension 2, starting from where one can get to $\ast,$ making on the
way some choice of the trajectory; $\bar{U}=\left\{  \bar{x}:x_{A}%
=x_{B},x_{AB}=\text{ }x_{BA},x_{O}=1-x_{A}-x_{B}-x_{AB}-\text{ }%
x_{BA}\right\}  .$

\item There is a cycle $\mathcal{C}\subset\left(  \mathbb{R}^{5}\right)
^{+},$ such that if $\bar{x}\left(  0\right)  \in\mathcal{C},$ then $\bar
{x}\left(  t\right)  \in\mathcal{C}$ for all $t>0.$ For example, the point
$\left\{  \bar{x}:x_{A}=1,x_{B}=x_{AB}=\text{ }x_{BA}=x_{O}=0\right\}  $
belongs to it. All the points of the cycle $\mathcal{C}$ are uniqueness
points. By $T_{\mathcal{C}}$ we will denote the time to go around the cycle
$\mathcal{C}$ once. We will now describe, just applying the definitions in a
straightforward way, the cycle $\mathcal{C},$ started from $\bar{x}\left(
0\right)  =\left\{  x_{A}=1,x_{B}=x_{AB}=\text{ }x_{BA}=x_{O}=0\right\}  .$
The first part of it happens for $t\in\left[  0,\frac{1}{9}\right]  ,$ during
which time the component $x_{A}\left(  t\right)  $ decays linearly from the
value $1$ to $0.$ The component $x_{B}\left(  t\right)  $ grows linearly, and
$x_{B}\left(  \frac{1}{9}\right)  =\frac{1}{9},$ the component $x_{AB}\left(
t\right)  $ grows linearly, and $x_{AB}\left(  \frac{1}{9}\right)  =\frac
{8}{9},$ while $x_{BA}$ and $x_{O}$ stay zero. On the next segment,
$t\in\left[  \frac{1}{9},1\right]  ,$ the component $x_{AB}\left(  t\right)  $
decays linearly, with rate $1,$ and so it vanishes at $t=1.$ The component
$x_{B}\left(  t\right)  $ continues to grow linearly, with the same speed, so
$x_{B}\left(  1\right)  =1.$ Other three components remain empty. So at time
$1$ we find ourselves accomplishing one half of the cycle. Therefore
$T_{\mathcal{C}}=2.$

\item If $\bar{x}\left(  0\right)  \notin\bar{U},$ then $\bar{x}\left(
t\right)  \in\mathcal{C}$ once $t\geq T,$ where $T=T\left(  \gamma_{O}%
,\gamma_{A},\gamma_{B},\gamma_{AB},\gamma_{BA}\right)  .$ We will prove this
claim under additional assumption that \textrm{dist}$\left(  \bar{x}\left(
0\right)  ,\mathcal{C}\right)  \equiv\inf_{y\in\mathcal{C}}\sum_{i}%
|x_{i}-y_{i}|<\varepsilon$ for some small $\varepsilon,$ since this will be
sufficient for our purposes.
\end{itemize}

\begin{lemma}
\label{LA} \textbf{Local attractor. }There exists an $\varepsilon>0$ such that
for any initial point $\bar{x}\left(  0\right)  $ with \textrm{dist}$\left(
\bar{x}\left(  0\right)  ,\mathcal{C}\right)  <\varepsilon$ we have $\bar
{x}\left(  t\right)  \in\mathcal{C}$ once $t>T=T\left(  \gamma_{O},\gamma
_{A},\gamma_{B},\gamma_{AB},\gamma_{BA}\right)  .$
\end{lemma}

\begin{proof}
Since for every point $\bar{c}$ on $\mathcal{C}$ we have either $c_{AB}=0$ or
$c_{BA}=0,$ we see that at least one of the coordinates $x_{AB}\left(
0\right)  $ or $x_{BA}\left(  0\right)  $ has to be less than $\varepsilon$.
Let us start with the case that both of them are positive, and we can then
assume that $x_{AB}\left(  0\right)  <\varepsilon,$ $x_{BA}\left(  0\right)
\geq x_{AB}\left(  0\right)  .$ Then for a short initial segment of time,
$\left[  0,t_{1}\right]  ,$ both coordinates $x_{AB}\left(  0\right)  $ and
$x_{BA}\left(  0\right)  $ decay with the same constant rate $2,$ until
$x_{AB}$ will vanish, which will happen at the moment $t_{1}=\frac
{x_{AB}\left(  0\right)  }{2}<\frac{\varepsilon}{2}.$ Also, for every point
$\bar{c}$ on $\mathcal{C}$ we have $c_{O}=0.$ Therefore $x_{O}\left(
0\right)  <\varepsilon,$ hence $x_{O}\left(  t_{1}\right)  <\varepsilon
+\frac{\varepsilon}{2}=\frac{3\varepsilon}{2},$ since the total arrival rate
to $\bar{O}$ does not exceed $4$, while its service rate equals $3$.

Now consider the case when $x_{BA}\left(  t_{1}\right)  >0.$ While
$x_{BA}\left(  t\right)  $ keeps being positive, \textbf{ }there is no flow
from $\bar{A}$ to $\bar{B}$ and, as a result, no flow from $\bar{B}$ to
$\bar{O}$. So if $x_{BA}\left(  t\right)  $ stays positive for $t\in\left[
t_{1},t_{2}=t_{1}+x_{O}\left(  t_{1}\right)  \right]  ,$ then $x_{O}\left(
t\right)  $ becomes $0$ at $t=t_{2}$ -- since it decays with rate $1$ when
$x_{AB}=0$ and $x_{BA}>0$ -- and will stay $0$ while $x_{BA}\left(  t\right)
$ is positive. The coordinate $x_{AB}\left(  t\right)  $ stays zero as well.
When finally $x_{BA}\left(  t\right)  $ vanishes for the first time, at
$t_{3}\geq t_{2},$ the value $x_{B}\left(  t_{3}\right)  $ has to be already
zero. Which means that $x_{A}\left(  t_{3}\right)  =1,$ so $\bar{x}\left(
t_{3}\right)  \in\mathcal{C}.$

In the remaining case, the first point, $t_{3},$ where the coordinate
$x_{BA}\left(  t\right)  $ vanishes, belongs to the segment $\left[
t_{1},t_{2}=t_{1}+x_{O}\left(  t_{1}\right)  \right]  .$ (That case contains
the situation when $x_{BA}\left(  t_{1}\right)  =0.$) Since the time
$t_{3}\leq2\varepsilon,$ and since initially the point $\bar{x}\left(
0\right)  $ was close to $\mathcal{C},$ the same is true for the point
$\bar{x}\left(  t_{3}\right)  .$ Since $x_{AB}\left(  t_{3}\right)
=x_{BA}\left(  t_{3}\right)  =0,$ we conclude that one of the two coordinates
-- either $x_{A}\left(  t_{3}\right)  ,$ or $x_{B}\left(  t_{3}\right)  $ --
should be $\varepsilon$-close to $1,$ while the remaining one, as well as
$x_{O}\left(  t_{3}\right)  ,$ should be $\varepsilon$-close to $0.$ Suppose
$x_{A}\left(  t_{3}\right)  \sim1.$ Note that the evolution of the point
$\bar{x}\left(  t_{3}\right)  $ is not uniquely defined if both $x_{A}\left(
t_{3}\right)  \ $and $x_{B}\left(  t_{3}\right)  $ are positive. As was
explained above, there are, in fact, three options\textbf{ }to choose from:

$i)$ If the choice is that the fluid $A$ \textquotedblleft goes
first\textquotedblright, then after a small time the coordinate $x_{O}$ will
vanish, and after the time of order $\frac{1}{9}$ the coordinate $x_{A}$ will
vanish as well, and we find ourselves on $\mathcal{C}.$

$ii)$ If the fluid $B$ \textquotedblleft goes first\textquotedblright, then
after a small time (of order $\varepsilon$) first $x_{B},$ and then $x_{BA}$
will vanish, and the fluid $x_{A}$ -- which is in the amount of order $1,$
will start to decay, so we find ourselves in the situation just considered.

$iii)$ The remaining option is when both fluids $A$ and $B$ \textquotedblleft
go simultaneously\textquotedblright\ into, resp., $AB$ and $BA$, at the rate
$\frac{10}{6}.$ As was explained above, the result will be that the levels of
the fluids $A$ and $B$ will decay with the rate $\frac{10}{6}-\frac{3}%
{2}=\frac{1}{6},$ while the level $x_{O}$ will be correspondingly raising. But
no later than the time $6\varepsilon$ the level $x_{B}$ will vanish, and since
$x_{AB}$ and $x_{BA}$ both were already zero, we again are in the situation
considered above, with $x_{O}$ being of order $\varepsilon.$ In fact, at any
time when the system is in phase $iii)$, it has three options: to pass
(forever!) to the phase $i)$ or $ii)$, or to stay in phase $iii)$. The above
arguments in $i)-iii)$ stays valid in this case and we get the required assertion.
\end{proof}

\textbf{Warning. }Our analysis shows that the time $T$ needed to reach the
cycle does not vanish with $\varepsilon$.

\begin{itemize}
\item \textbf{Continuity. }The above proof implies the following (weaker)
substitute for the property of the continuous dependence of the trajectory on
the initial condition. Let $\bar{x}\left(  0\right)  $ and $\bar{c}\left(
0\right)  $ be two initial points, and suppose that $\left\vert \bar{x}\left(
0\right)  -\bar{c}\left(  0\right)  \right\vert <\varepsilon,$ and $\bar
{c}\left(  0\right)  \in\mathcal{C}.$ Then there exists a constant $C,$ such
that for any $t$ and for any version of the $x$-trajectory we have $\left\vert
\bar{x}\left(  t\right)  -\bar{c}\left(  t\right)  \right\vert <C\varepsilon.$
The condition $\bar{c}\left(  0\right)  \in\mathcal{C}$ is, evidently,
crucial; without it our statement fails.

\item The cycle $\mathcal{C}$ depends on $\gamma_{O},\gamma_{A},\gamma
_{B},\gamma_{AB},\gamma_{BA},$ and is non-trivial for our choice of these
parameters. For some other values of $\gamma_{O},\gamma_{A},\gamma_{B}%
,\gamma_{AB},\gamma_{BA}$ it is reduced to the point $\ast,$ which then is a
stable fixed point.

\item All the above properties of our system would still be valid if we
perturb slightly the vector $\bar{\gamma}=\left\{  \gamma_{O},\gamma
_{A},\gamma_{B},\gamma_{AB},\gamma_{BA}\right\}  $ of the parameters around
the point $\left\{  3,10,10,2,2\right\}  $ of our choice (even if the
perturbation does not respect the symmetry $\gamma_{A}=\gamma_{B},$
$\gamma_{AB}=\gamma_{BA}$).

\item Let $\bar{x}\left(  t\right)  \subset\mathcal{C}$ be a cyclic
trajectory. Let us denote by $\bar{\lambda}^{\mathcal{C}}\left(  t\right)  $
the corresponding (periodic) function of the inflows. The stability property
just formulated implies immediately that our system in the cyclic regime is
underloaded. In other words, for $\left\{  \gamma_{O},\gamma_{A},\gamma
_{B},\gamma_{AB},\gamma_{BA}\right\}  =\left\{  3,10,10,2,2\right\}  $ there
exists $\delta>0$ such that we have
\begin{equation}
\frac{1}{\gamma_{O}}\int_{0}^{T_{\mathcal{C}}}\lambda_{O}^{\mathcal{C}}\left(
t\right)  dt<\left(  1-\delta\right)  T_{\mathcal{C}},\label{0013}%
\end{equation}%
\begin{equation}
\frac{1}{\gamma_{A}}\int_{0}^{T_{\mathcal{C}}}\lambda_{A}^{\mathcal{C}}\left(
t\right)  dt+\frac{1}{\gamma_{BA}}\int_{0}^{T_{\mathcal{C}}}\lambda
_{BA}^{\mathcal{C}}\left(  t\right)  dt<\left(  1-\delta\right)
T_{\mathcal{C}},\label{0015}%
\end{equation}

\noindent and the same relation for the $\bar{B}$ node.
\end{itemize}

We now want to consider the \textquotedblleft open\textquotedblright\ system
$\Delta^{o},$ which is obtained from $\Delta$ by the following construction:
every exiting fluid now goes not to the corresponding node, but leaves the
system. On the other hand, there is some inflow, $\bar{\lambda}\left(
t\right)  $, entering the system from the outside.

For the future use we introduce now the following compact subset
$\mathcal{K}\equiv\mathcal{K}\left(  \bar{\gamma}\right)  \subset\left(
\mathbb{R}^{5}\right)  ^{+}:$%
\[
\mathcal{K}=\left\{  \bar{x}\in\left(  \mathbb{R}^{5}\right)  ^{+}%
:\max\left\{  \frac{1}{2}x_{AB}+\frac{1}{10}x_{B},\frac{1}{2}x_{BA}+\frac
{1}{10}x_{A},\frac{1}{3}x_{O}\right\}  <10\right\}  .
\]

\begin{lemma}
\label{l3} Let the open system $\Delta^{o}$ be in the state $\bar{x}\left(
0\right)  \in\left(  \mathbb{R}^{5}\right)  ^{+}.$ Consider the functional $L$
on $\mathbb{R}^{5},$ given by
\begin{equation}
L\left(  \bar{x}\right)  =\left\{
\begin{array}
[c]{cc}%
0 & \text{ if }\bar{x}\in\mathcal{K},\\
\max\left\{  \frac{1}{2}x_{AB}+\frac{1}{10}x_{B},\frac{1}{2}x_{BA}+\frac
{1}{10}x_{A},\frac{1}{3}x_{O}\right\}  & \text{ otherwise.}%
\end{array}
\right.  \label{037}%
\end{equation}
Suppose that $L\left(  \bar{x}\left(  0\right)  \right)  >10.$ Then there
exists a constant $C>0$, such that for all external flow rates $\left\{
\bar{\lambda}\left(  t\right)  ,t\in\left[  0,T_{\mathcal{C}}\right]
\right\}  $ which are close enough to $\left\{  \bar{\lambda}^{\mathcal{C}%
}\left(  t\right)  ,t\in\left[  0,T_{\mathcal{C}}\right]  \right\}  $ in the
$L^{1}$ distance, we have
\[
L\left(  \bar{x}\left(  0\right)  \right)  -L\left(  \bar{x}\left(
T_{\mathcal{C}}\right)  \right)  >C.
\]

\end{lemma}

\begin{proof}
For the case of the inflows with rates $\bar{\lambda}^{\mathcal{C}}\left(
t\right)  $ our statement follows from the underload property, due to the
relations (\ref{0013})-(\ref{0015}). Therefore it holds for inflows that are
close enough, by continuity.
\end{proof}

\subsection{$M$ coupled fluid networks \label{S23}}

Let $\Delta_{M}$ be the dynamical system on $\left(  \mathbb{R}^{5M}\right)
^{+},$ obtained from $M$ copies of $\Delta,$ interconnected in the mean-field
manner, as follows.\textbf{ } Each node $\bar{O}_{i},$ $i=1,...,M$, is
connected to \textit{all }of the nodes $\bar{A}_{j},$ $\bar{B}_{j},$
$j=1,...,M$, and its fluid, of the amount $x_{O,i},$ flows into the nodes
$\bar{A}_{j},$ $\bar{B}_{j}$ in equal amounts. The rate of each of these flows
is now $M^{-1}\frac{\gamma_{O}}{2}=\frac{3}{2M}$, which means, as before, that
three units of the fluid $x_{O,i}$ leave $\bar{O}_{i}$ per unit time, so each
of the set $\left\{  \bar{A}_{j}\right\}  $ and $\left\{  \bar{B}_{j}\right\}
$ gets $\frac{3}{2}$ units of incoming fluids, $A$ and $B,$ per unit time. In
a similar way, the fluid $A$ from every node $\bar{A}_{i}$ is splitted among
all nodes $\left\{  \bar{B}_{j}\right\}  ,$ so the rate of every individual
flow $\bar{A}_{i}\rightarrow\bar{B}_{j}$ is $M^{-1}\gamma_{A}=\frac{10}{M},$
and so on. The priorities are kept the same: if the node $\bar{A}_{i},$ say,
is in the state with both amounts $x_{A,i}$ and $x_{BA,i}$ positive, then
\textit{the fluid }$BA$\textit{ goes first}.

Again, we first describe the open network, i.e. the Non-Homogeneous Dynamical
System. We will not need the general case here; it is enough for us to
consider the net inflow defined by the same function $\bar{Y}(t)\in\left(
\mathbb{R}^{5}\right)  ^{+}$ as in the previous subsection; every node then
gets $\frac{1}{M}$-th part of the inflow, so, for example, for each
$i=1,...,M$ the net inflow function of the node $\bar{O}_{i}$ equals to
$\frac{1}{M}Y_{O}\left(  t\right)  .$ Once we are also given the initial
values $\bar{x}\in\left(  \mathbb{R}^{5M}\right)  ^{+}$ of the fluid levels,
the net outflows $Z_{a,i}^{\bar{x},\bar{Y}}\left(  t\right)  ,$ $a\in\left\{
O,A,B,AB,BA\right\}  ,$ $i=1,...,M,$ are defined as above, see $\left(
\ref{z}\right)  .$

Passing to the closed system, instead of relations (\ref{020})-(\ref{022}), we
impose the relations\textbf{ }%
\begin{equation}
Y_{O}\left(  t\right)  =\sum_{i=1}^{M}(Z_{AB,i}^{\bar{x},\bar{Y}}\left(
t\right)  +Z_{BA,i}^{\bar{x},\bar{Y}}\left(  t\right)  ), \label{N020}%
\end{equation}%
\begin{equation}
Y_{A,j}\left(  t\right)  =\sum_{i=1}^{M}\frac{1}{2}Z_{O,i}^{\bar{x},\bar{Y}%
}\left(  t\right)  , \label{N021}%
\end{equation}%
\begin{equation}
Y_{AB,j}\left(  t\right)  =\sum_{i=1}^{M}Z_{A,i}^{\bar{x},\bar{Y}}\left(
t\right)  . \label{N022}%
\end{equation}
Again, all the functions $Y_{\ast}(t)$ and $Z_{\ast}(t)$ are Lipschitz
continuous and, hence, differentiable almost everywhere. These derivatives
will be denoted, again, by $y_{\ast}(t)$ and $z_{\ast}(t)$. For each node
(say, $\bar{A}_{i}$) of $\Delta_{M}$, the evolution of the amount of fluids
$x_{A,i}(t)$ and $x_{BA,i}(t)$ is found from the corresponding inflows
$\frac{1}{M}Y_{A}(t)$ and $\frac{1}{M}Y_{BA}(t)$ and initial states
$x_{A,i}(0)$ and $x_{BA,i}(0)$ in the same manner as for the network $\Delta$.

We will consider trajectories $\bar{x}^{M}\left(  t\right)  \in\left(
\mathbb{R}^{5M}\right)  ^{+}$ with $\left\vert \bar{x}^{M}\left(  t\right)
\right\vert =M,$ i.e. we have the unit amount of fluid per elementary system
$\Delta\subset\Delta_{M}.$

For every $M$ we will define now another dynamical system, acting on $\left(
\mathbb{R}^{5}\right)  ^{+}.$ This one, also denoted by $\Delta_{M},$ will be
of central importance for the present paper. However, it will be not the usual
dynamical system. It will be defined not as a group of transformations of
$\left(  \mathbb{R}^{5}\right)  ^{+},$ but directly on the (sub)set
$\mathcal{M}_{M}$ of some atomic probability measures on $\left(
\mathbb{R}^{5}\right)  ^{+}.$ This transformation will not be linear on
$\mathcal{M},$ and for that reason we will call it \textit{non-linear
dynamical system.} In fact, it is just a convenient representation of our
initial dynamical system on $\left(  \mathbb{R}^{5M}\right)  ^{+},$ suitable
for passing to the limit $M\rightarrow\infty.$ The construction is very simple:

To every point $\bar{x}^{M}=\left\{  \left(  x_{O}^{i},x_{\bar{A}}^{i}=\left(
x_{A}^{i},x_{BA}^{i}\right)  ,x_{\bar{B}}^{i}=\left(  x_{B}^{i},x_{AB}%
^{i}\right)  \right)  ,i=1,...,M\right\}  \in\left(  \mathbb{R}^{5M}\right)
^{+}$ we can assign a probability measure on $\left(  \mathbb{R}^{5}\right)
^{+}$ in the following way: we put $\mu_{O}=\frac{1}{M}\sum_{i=1}^{M}%
\delta_{x_{O}^{i}},$ $\mu_{\bar{A}}=\frac{1}{M}\sum_{i=1}^{M}\delta
_{x_{\bar{A}}^{i}},$ $\mu_{\bar{B}}=\frac{1}{M}\sum_{i=1}^{M}\delta
_{x_{\bar{B}}^{i}},$ and we define $\mu\equiv\mu_{\bar{x}^{M}}=\mu_{O}%
\times\mu_{\bar{A}}\times\mu_{\bar{B}}\in\mathcal{M}_{M}.$ Now, if the point
$\bar{x}^{M}$ evolves according to $\Delta_{M},$ so is the measure $\mu;$
moreover, the evolution $\mu\left(  t\right)  $ of $\mu$ is well defined and
does not depend on the choice of the preimage, so if $\mu_{\bar{x}_{1}^{M_{1}%
}}=\mu_{\bar{x}_{2}^{M_{2}}},$ then $\mu_{\bar{x}_{1}^{M_{1}}}\left(
t\right)  =\mu_{\bar{x}_{2}^{M_{2}}}\left(  t\right)  .$

The set of measures $\mathcal{M}_{M}$ on $\left(  \mathbb{R}^{5}\right)  ^{+}%
$consists of all measures $\mu,$ having the properties

\begin{enumerate}
\item $\mu$ is a product,
\begin{equation}
\mu\equiv\left(  \mu_{O},\mu_{\bar{A}},\mu_{\bar{B}}\right)  \equiv\mu
_{O}\times\mu_{\bar{A}}\times\mu_{\bar{B}}\equiv\Pi_{\bar{o}}\left[
\mu\right]  \times\Pi_{\bar{A}}\left[  \mu\right]  \times\Pi_{\bar{B}}\left[
\mu\right]  , \label{0012}%
\end{equation}
of probability measures on $\mathbb{R}^{1}=\left\{  x_{O}\right\}  ,$ resp.
$\mathbb{R}^{2}=\left\{  x_{A},x_{BA}\right\}  $ and $\mathbb{R}^{2}=\left\{
x_{B},x_{AB}\right\}  .$ Here we denote by $\Pi_{\ast}$-s the various
projections (or marginals),

\item we have $\int\left\vert \bar{x}\right\vert d\mu=\int x_{O}d\mu_{O}%
+\int\left(  x_{A}+x_{BA}\right)  d\mu_{\bar{A}}+\int\left(  x_{B}%
+x_{AB}\right)  d\mu_{\bar{B}}=1,$

\item we have $\mu_{O}=\frac{1}{M}\sum_{i=1}^{M}\delta_{\bar{x}_{i}}$ for some
(not necessarily distinct) $\bar{x}_{i}\in\mathbb{R}^{1},$ $i=1,...,M,$
likewise $\mu_{\bar{A}}=\frac{1}{M}\sum_{i=1}^{M}\delta_{\bar{x}_{i}^{\prime}%
},$ $\mu_{\bar{B}}=\frac{1}{M}\sum_{i=1}^{M}\delta_{\bar{x}_{i}^{\prime\prime
}},$ $\bar{x}_{i}^{\prime},\bar{x}_{i}^{\prime\prime}\in\mathbb{R}^{2}.$
\end{enumerate}

Properties of $\Delta_{M}:$

\begin{itemize}
\item The set of the fixed points of $\Delta_{M}$ consists of measures
$\mu=\Pi_{\bar{o}}\left[  \mu\right]  \times\delta_{x_{\bar{A}}=0}\times
\delta_{x_{\bar{B}}=0},$ where $\delta_{x_{\bar{A}}=0}$ and $\delta
_{x_{\bar{B}}=0}$ are unit atoms at the origin, while $\Pi_{\bar{o}}\left[
\mu\right]  $ is the projection on the coordinate $x_{O}.$ In words, that
means that all the fluid stays permanently in the $O$-nodes (in arbitrary
amounts, adding up to $M$). This fact follows from Proposition \ref{PP1}.

\item If $\mu\left(  0\right)  =\delta_{\bar{x}}\in\mathcal{M}_{M},$ then
$\mu\left(  t\right)  =\delta_{\bar{x}\left(  t\right)  },\ $where $\bar
{x}\left(  t\right)  $ is the trajectory of $\Delta$ with $\bar{x}\left(
0\right)  =$ $\bar{x}.$ In particular if $\bar{x}\in\mathcal{C},$ then
$\bar{x}\left(  t\right)  \in\mathcal{C}.$

\item Note that the dynamics $\Delta_{M}$ on $\mathcal{M}_{M}$ is
\textquotedblleft non-linear\textquotedblright, in the sense that in general
in the situation when $\mu\left(  0\right)  =\alpha\mu^{\prime}\left(
0\right)  +\left(  1-\alpha\right)  \mu^{\prime\prime}\left(  0\right)
,~~\mu\left(  0\right)  ,\mu^{\prime}\left(  0\right)  ,\mu^{\prime\prime
}\left(  0\right)  \in\mathcal{M}_{M},~~0<\alpha<1,$ we have $\mu\left(
t\right)  \neq\alpha\mu^{\prime}\left(  t\right)  +\left(  1-\alpha\right)
\mu^{\prime\prime}\left(  t\right)  $ for $t>0.$ (There is nothing strange or
unusual in this relation, since the dynamical system in question is itself
defined on the space of measures as its state space, and not on the $\left(
\mathbb{R}^{5}\right)  ^{+}$.)

\item Let $\rho_{KROV}$ be the Kantorovich-Rubinstein-Ornstein-Vaserstein
distance on the probability measures on $\left(  \mathbb{R}^{5}\right)  ^{+},$
corresponding to the metric $\rho\left(  \bar{x},\bar{y}\right)  =\sum
_{i=1}^{5}\left\vert x_{i}-y_{i}\right\vert $ on $\left(  \mathbb{R}%
^{5}\right)  ^{+}.$ \textbf{ }{\small (We recall briefly, that if }$\mu,$%
$\mu^{\prime}${\small are two probability measures on a metric space }$\left(
X,\rho\right)  ${\small , then }$\rho_{KROV}\left(  \mu,\mu^{\prime}\right)
=\inf_{\kappa}\int\rho\left(  x,x^{\prime}\right)  ~d\kappa\left(
x,x^{\prime}\right)  ,${\small where the }$\inf$ {\small is taken over all
probability measures }$\kappa$ {\small on }$X\times X,$ {\small such that
}$\kappa\left(  A\times X\right)  =\mu\left(  A\right)  ,$ $\kappa\left(
X\times A\right)  =\mu^{\prime}\left(  A\right)  .${\small )} Suppose that the
initial measure $\mu\left(  0\right)  $ is close enough to the cycle
$\mathcal{C},$ which means that for some $x\in\mathcal{C}$ we have
\begin{equation}
\rho_{KROV}\left(  \mu\left(  0\right)  ,\delta_{x}\right)  <\varepsilon
.\label{052}%
\end{equation}
Let $Y\left(  t\right)  \in\mathcal{Y}\left(  \mu\left(  0\right)  \right)  $
be one of the possible net inflows, corresponding to the initial state
$\mu\left(  0\right)  ,$ and $\mu\left(  t\right)  $ be the corresponding
evolution. Then there exists the time $T=T\left(  M,\gamma_{O},\gamma
_{A},\gamma_{B},\gamma_{AB},\gamma_{BA}\right)  ,$ such that for all $t\geq T$
we have $\mu\left(  t\right)  =\delta_{\bar{x}\left(  t\right)  }$ with
$\bar{x}\left(  t\right)  \in\mathcal{C}\mathbf{.}$ Moreover, let us define
the compact $\Lambda_{K}$ by%
\[
\Lambda_{K}=\left\{  \bar{x}=\left(  x_{O},x_{A},...,x_{B}\right)
:x_{O}<K,x_{A}<K,...,x_{B}<K\right\}  ,
\]
and let $\mu\Bigm|_{K}\left(  t\right)  $ be the evolution of the restriction
$\mu\left(  0\right)  \Bigm|_{K}$ under the same evolution, defined by the
flow rates $Y\left(  t\right)  \in\mathcal{Y}\left(  \mu\left(  0\right)
\right)  .$ (Once the inflows $Y\left(  t\right)  $ are fixed, the evolution
becomes the usual (non-autonomous) \textit{linear }dynamical system, so we can
apply the dynamics to the summand $\mu\left(  0\right)  \Bigm|_{K}$ of the
measure $\mu\left(  0\right)  .$) We choose $K$ to be large enough, so that
from $\rho_{KROV}\left(  \mu\left(  0\right)  ,\delta_{x}\right)
<\varepsilon$ for some $x\in\mathcal{C}$ it follows that
\begin{equation}
\mu\left(  0\right)  \left[  K\right]  >1-\varepsilon.\label{053}%
\end{equation}
We now claim the following:
\end{itemize}

\begin{lemma}
\label{Pr4} Under conditions $\left(  \ref{052}\right)  $ and $\left(
\ref{053}\right)  ,$ there exists the time moment $T^{\prime}=T^{\prime
}\left(  K,\gamma_{O},\gamma_{A},\gamma_{B},\gamma_{AB},\gamma_{BA}\right)  ,$
such that for every $t\geq T^{\prime}$ (\textit{uniformly in }$M$ !) there
exists a point $\bar{x}\left(  t\right)  ,$ such that $\mu\Bigm|_{K}\left(
t\right)  $ is just the atom at that point: $\mu\Bigm|_{K}\left(  t\right)
=c\delta_{\bar{x}\left(  t\right)  }.$ Moreover, \textrm{dist}$\left(  \bar
{x}\left(  t\right)  ,\mathcal{C}\right)  <\tilde{c},$ and  $c=\mu\left(
0\right)  \left[  K\right]  \rightarrow1$ while $\tilde{c}\rightarrow0$ as
$K\rightarrow\infty.$
\end{lemma}

\textbf{Notes.}

\begin{itemize}
\item We will prove only the statement about the existence of the time moment
$T^{\prime}$, since below we will not use the time moment $T\left(
M,\gamma_{O},\gamma_{A},\gamma_{B},\gamma_{AB},\gamma_{BA}\right)  .$

\item Our proof can be extended \textit{literally} to the case $M=\infty$ of
the next Subsection.
\end{itemize}

\begin{proof}
Note first that due to the mean-field nature of our graph, the flows to all
the $\bar{A}$-nodes are equal at every time moment (as well as to all the
$\bar{B}$-nodes or $\bar{O}$-nodes). Consider now any subset $Q$ of the
$\bar{A}$-nodes, $\left\vert Q\right\vert \leq M,$ and let $I_{BA}$ be the
index, for which $\left(  x_{BA}\right)  _{I_{BA}}\left(  0\right)
\geq\left(  x_{BA}\right)  _{i}\left(  0\right)  $ for all $i\in Q.$ Then,
clearly, this relation holds at later moments, i.e. $\left(  x_{BA}\right)
_{I_{BA}}\left(  t\right)  \geq\left(  x_{BA}\right)  _{i}\left(  t\right)  .$
In the same way, define the index $I_{A}$ as the one, for which $\left(
x_{A}\right)  _{I_{A}}\left(  0\right)  \geq\left(  x_{A}\right)  _{i}\left(
0\right)  .$ Then, if all the variables $\left(  x_{BA}\right)  _{i}\left(
0\right)  $ are equal for $i\in Q,$ the relation $\left(  x_{A}\right)
_{I_{A}}\left(  t\right)  \geq\left(  x_{A}\right)  _{i}\left(  t\right)  $
holds at all later moments. We will use this property for the set $Q=\left\{
i\right\}  $ of indices, which satisfy $\left(  x_{BA}\right)  _{i}\left(
0\right)  <K,$ $\left(  x_{A}\right)  _{i}\left(  0\right)  <K.$

Let us show that there exists the time moment, such that before it every node
in $Q$ will be empty for some time duration. In view of what was said before,
it means that after that time moment all the nodes in $Q$ will be
synchronized. To see this depletion, note that the initial supply of the fluid
$BA$ at any node $i\in Q$ is not exceeding $K,$ so it will be over before
$t^{\prime}=\frac{K}{2}.$ Next, there exists a moment $t^{\prime\prime},$
after which $\frac{99}{100}$ (say) of \textquotedblleft
atoms\textquotedblright\ of the heavy fluid $BA$, passing through our node
$i\in Q$ were at earlier moments at some $\bar{O}$-node. Indeed, another
option would be that such an atom was staying at some $\bar{B}$ node for all
the time duration $t^{\prime\prime}.$ That means that the initial total amount
of fluid at this node was very high, once $t^{\prime\prime}$ is chosen to be
large. However, the proportion of such nodes has to be small, due to the
simple fact that the total amount of fluid per node is of the order of one.
But the rate, at which the fluid goes from the $\bar{O}$ node to the $\bar{B}$
node is never higher than $\frac{3}{2}.$ Let $K^{\prime\prime}$ be the amount
of fluid at our node $i$ at the moment $t^{\prime\prime}.$ Clearly it is at
most $12t^{\prime\prime}.$ Suppose the node is not empty during the time
interval $\left[  t^{\prime\prime},t^{\prime\prime}+T\right]  .$ Than it works
all the time at full capacity. Let $0\leq k\left(  t\right)  \leq1$ be the
fraction at moment $t$ of the capacity of the node, used by the heavy fluid,
while the remaining fraction $1-k\left(  t\right)  $ is used by the light
fluid. Then the amount of heavy fluid, which left the server during this time
interval, is $2\int_{t^{\prime\prime}}^{t^{\prime\prime}+T}k\left(  t\right)
dt,$ while the corresponding amount of the light fluid is $10\int
_{t^{\prime\prime}}^{t^{\prime\prime}+T}\left(  1-k\left(  t\right)  \right)
dt.$ Since the light fluid flows into the node with the rate at most $\frac
{3}{2},$ we have that the relation
\[
10\int_{t^{\prime\prime}}^{t^{\prime\prime}+T}\left(  1-k\left(  t\right)
\right)  dt\leq\frac{3}{2}T+K^{\prime\prime}%
\]
has to hold, since the amount of light fluid, leaving the node, can not exceed
the initial amount present at the node plus the amount which came to the node
during the time interval $T.$ For the heavy fluid we similarly have
\[
2\int_{t^{\prime\prime}}^{t^{\prime\prime}+T}k\left(  t\right)  dt\leq\frac
{3}{2}T+10\left(  \frac{1}{100}T\right)  +K^{\prime\prime}.
\]
The two relations imply that
\[
10T\leq\left(  9\frac{1}{2}\right)  T+6K^{\prime\prime}.
\]
So $T\leq12K^{\prime\prime},$ which establish our depletion claim for the
$\bar{A}$ (as well as for $\bar{B}$) nodes at some moment $t^{\prime
\prime\prime},$ independent of $M$ and $\varepsilon,$ provided only that
$\varepsilon$ is small.

Thus far we were not using the condition $\rho_{KROV}\left(  \mu\left(
0\right)  ,\delta_{x}\right)  <\varepsilon,$ without which our claim about the
existence of the time moment $T^{\prime}\left(  K,\gamma_{O},\gamma_{A}%
,\gamma_{B},\gamma_{AB},\gamma_{BA}\right)  $ is not valid -- see, for
example, the first property of the dynamics $\Delta_{M}$. We will use it now,
in dealing with the $\bar{O}$ nodes. Due to the above discussion and the
continuity property, our statement is reduced to the following one: consider
the initial measure $\mu\left(  0\right)  ,$ having the properties:
$\rho_{KROV}\left(  \mu\left(  0\right)  ,\delta_{\bar{x}}\right)
<\varepsilon^{\prime}(=C\varepsilon,$ see Continuity Property of the previous
section), while in the decomposition $\mu\left(  0\right)  =\mu\left(
0\right)  _{\bar{O}}\times\mu\left(  0\right)  _{\bar{A}}\times\mu\left(
0\right)  _{\bar{B}}$ we have $\mu\left(  0\right)  _{\bar{A}}=\left(
1-\varepsilon\right)  \delta_{\tilde{x}_{\bar{A}}}+\varkappa_{\bar{A}},$
$\mu\left(  0\right)  _{\bar{B}}=\left(  1-\varepsilon\right)  \delta
_{\tilde{x}_{\bar{B}}}+\varkappa_{\bar{B}},$ with the vectors $\tilde{x}%
_{\bar{A}},\tilde{x}_{\bar{B}}$ $\in\mathbb{R}^{2}$ close to the corresponding
projections $\left(  x_{A},x_{BA}\right)  ,$ resp. $\left(  x_{B}%
,x_{AB}\right)  $ of the vector $\bar{x}.$ But that means that the flows
into\textbf{ }the $\bar{O}$ nodes will be almost always almost equal to these
on the cycle $\mathcal{C},$ so in finite time all of the $\bar{O}$-nodes which
initially have their levels $\leq K$ will become empty. Indeed, on the cycle
the flow to the $\bar{O}$ node has rate $2,$ while the capacity of these nodes
equals $3.$

Summarizing, we have thus far that after a finite time, independent of $M,$
all the nodes in $Q$ are synchronized, and moreover all the $\bar{O}$-nodes in
$Q$ are empty. The application of the Continuity Property and the Attraction
Lemma \ref{LA} finishes the proof.
\end{proof}

\subsection{$M\rightarrow\infty$ fluid network}

This is again a dynamical system, $\Delta_{\infty},$ acting on probability
measures $\mathcal{M}=\mathcal{M}\left(  \left(  \mathbb{R}^{5}\right)
^{+}\right)  $ on $\left(  \mathbb{R}^{5}\right)  ^{+}.$ One way of defining
it is to say that the family $\mu\left(  t\right)  \in\mathcal{M}$ is a
trajectory of $\Delta_{\infty},$ iff for any $M$ there is a trajectory
$\mu_{M}\left(  t\right)  \in\mathcal{M}_{M}$ of $\Delta_{M},$ so that for
every $t$ we have $\mu_{M}\left(  t\right)  \rightarrow\mu\left(  t\right)  $ weakly.

Now we will give another description of $\Delta_{\infty},$ which does not make
use of the limit $M\rightarrow\infty.$ As was the case with the dynamics
$\Delta$ and $\Delta_{M},$ we will define first the NHDS version of
$\Delta_{\infty}.$ Let the measure $\mu_{0}\in\mathcal{M}\left(  \left(
\mathbb{R}^{5}\right)  ^{+}\right)  $ and let the function $\bar{Y}\left(
t\right)  \in\left(  \mathbb{R}^{5}\right)  ^{+}$ be given, which is the net
inflow of our fluids. Then the corresponding evolution $\mu=\left\{  \mu
_{t}\in\mathcal{M}\left(  \left(  \mathbb{R}^{5}\right)  ^{+}\right)
\right\}  $ of the state $\mu_{0}$ is defined to be just the evolution of the
measure $\mu_{0}$ under the dynamics $\Delta\left(  \bar{Y}\right)  .$ Again,
the net outflows $\bar{Z}^{\bar{x},\bar{Y}}\left(  t\right)  ,$ $\bar{x}%
\in\left(  \mathbb{R}^{5}\right)  ^{+}$ are given by $\left(  \ref{z}\right)
.$

Now we can define the evolution $\mu=\left\{  \mu_{t}\in\mathcal{M}\left(
\left(  \mathbb{R}^{5}\right)  ^{+}\right)  \right\}  $ of the initial measure
$\mu_{0}$ under Non-Linear Dynamical System $\Delta_{\infty}$ (NLDS) as the
NHDS evolution of it under any dynamics $\Delta\left(  \bar{Y}\right)  ,$ with
$\bar{Y}$ satisfying the equations:%
\begin{align}
Y_{O}\left(  t\right)   &  =\int(Z_{AB}^{\bar{x},\bar{Y}}\left(  t\right)
+Z_{BA}^{\bar{x},\bar{Y}}\left(  t\right)  )d\mu_{0}\left(  \bar{x}\right)
,\nonumber\\
Y_{A}\left(  t\right)   &  =\frac{1}{2}\int Z_{O}^{\bar{x},\bar{Y}}\left(
t\right)  d\mu_{0}\left(  \bar{x}\right)  ,\label{031}\\
Y_{AB}\left(  t\right)   &  =\int Z_{A}^{\bar{x},\bar{Y}}\left(  t\right)
d\mu_{0}\left(  \bar{x}\right)  ,\nonumber
\end{align}
and symmetric relations for $B$ and $BA$ variables. The set of all such flows
$\bar{Y}\left(  t\right)  $ will be denoted by\textbf{ }$\mathcal{Y}\left(
\mu_{0}\right)  .$

We are calling the dynamical system $\Delta_{\infty}$ non-linear, since in
general we will have for $\mu_{0}=\frac{1}{2}\left(  \mu_{0}^{\prime}+\mu
_{0}^{\prime\prime}\right)  $ that $\mu_{t}\not =\frac{1}{2}\left(  \mu
_{t}^{\prime}+\mu_{t}^{\prime\prime}\right)  $ when $t>0.$ Note also that if
$\mu_{0}=\delta_{\bar{x}\left(  0\right)  }$ for some point $\bar{x}\left(
0\right)  ,$ then the family $\mu_{t}$ is a $\Delta_{\infty}$ trajectory iff
$\mu_{t}=\delta_{\bar{x}\left(  t\right)  },$ with $\bar{x}\left(  t\right)  $
being some $\Delta$-evolution of $\bar{x}\left(  0\right)  .$

Properties of $\Delta_{\infty}:$

\begin{enumerate}
\item the set of the fixed points of $\Delta_{\infty}$ consists of measures
$\mu=\Pi_{\bar{o}}\left[  \mu\right]  \times\delta_{x_{\bar{A}}=0}\times
\delta_{x_{\bar{B}}=0},$ where $\Pi_{\bar{o}}\left[  \mu\right]  $ is the
projection on the coordinate $x_{O}.$

\item if $\mu\left(  0\right)  \in\mathcal{M}_{M}$ for some $M,$ then the
$\Delta_{M}$-dynamics and $\Delta_{\infty}$-dynamics with that initial data
coincide. In particular, if $\mu\left(  0\right)  =\delta_{\bar{x}}%
\in\mathcal{M}$ with $\bar{x}\in\mathcal{C},$ then $\mu\left(  t\right)
=\delta_{\bar{x}\left(  t\right)  },\ $where $\bar{x}\left(  t\right)  $ is
the trajectory of $\Delta$ with $\bar{x}\left(  0\right)  =$ $\bar{x}.$
\end{enumerate}

For all our purposes it is sufficient to prove the following property of our NLDS.

\begin{proposition}
\label{P1} Let the measure $\mu=\mu\left(  0\right)  $ on $\left(
\mathbb{R}^{5}\right)  ^{+}$ have the following properties:

$i)$ Unit mass:
\begin{equation}
\int_{\left(  \mathbb{R}^{5}\right)  ^{+}}\left(  x_{A}+x_{B}+x_{BA}%
+x_{AB}+x_{O}\right)  d\mu=1, \label{030}%
\end{equation}

$ii)$ Exponential moment condition: for some $\alpha>0,$ $A<\infty$ we have
\begin{equation}
\left\langle \exp\left\{  \alpha L\left(  \bar{x}\right)  \right\}
\right\rangle _{\mu\left(  0\right)  }<3A\label{032}%
\end{equation}
(see $\left(  \ref{037}\right)  $),

$iii)$ For some $x\in\mathcal{C}$ we have
\begin{equation}
\rho_{KROV}\left(  \mu\left(  0\right)  ,\delta_{x}\right)  <\varepsilon,
\label{034}%
\end{equation}
with $\varepsilon$ small enough (depending on $\alpha$ and $A$).

Consider now some Non-Linear Dynamical System $\Delta_{\infty}$ (NLDS),
defined by the initial state $\mu\left(  0\right)  .$ In other words,
$\Delta_{\infty}=\Delta_{\infty}\left(  \bar{Y}\left(  \cdot\right)  \right)
,$ for some $\bar{Y}\left(  \cdot\right)  \in\mathcal{Y}\left(  \mu\left(
0\right)  \right)  .$ Then for $t\rightarrow\infty$ the evolving measure
$\mu\left(  t\right)  $ satisfies :

$I)$
\begin{equation}
\rho_{KROV}\left(  \mu\left(  t\right)  ,\delta_{z\left(  t\right)  }\right)
\rightarrow0 \label{035}%
\end{equation}
for appropriate $z\left(  t\right)  \in\mathcal{C}$ ,

$II)$
\begin{equation}
\left\langle \exp\left\{  \alpha L\left(  \bar{x}\right)  \right\}
\right\rangle _{\mu\left(  t\right)  }\rightarrow0. \label{036}%
\end{equation}
Moreover, the convergence in $\left(  \ref{035}\right)  ,$ $\left(
\ref{036}\right)  $ is uniform over all $\bar{Y}\left(  \cdot\right)
\in\mathcal{Y}\left(  \mu\left(  0\right)  \right)  $ and all initial measures
$\mu\left(  0\right)  $ satisfying $\left(  \ref{030}\right)  $-$\left(
\ref{034}\right)  .$
\end{proposition}

To establish it we first prove the following simpler fact.

\begin{proposition}
\label{P2} For every $T,\varepsilon$ there exists a value $\bar{\varepsilon
}\left(  T,\varepsilon\right)  ,$ such that the following holds:

$i)$ For every $T$ we have $\bar{\varepsilon}\left(  T,\varepsilon\right)
\rightarrow0$ as $\varepsilon\rightarrow0.$

$ii)$ Let $\mu$ be any measure satisfying $i)$ -- $iii)$ above. Then for every
$t\in\left[  0,T\right]  $ we have $\rho_{KROV}\left(  \mu\left(  t\right)
,\delta_{x\left(  t\right)  }\right)  <\bar{\varepsilon}\left(  T,\varepsilon
\right)  .$
\end{proposition}

\begin{proof}
Let the sequence of measures $\mu_{n}\left(  0\right)  $ converge to
$\delta_{x}$ in the $KROV$ metrics, and let it satisfy the conditions of the
Proposition \ref{P1}. Let us consider the set of trajectories $\mu_{n}\left(
t\right)  ,$ $0\leq t\leq T.$ The family of these trajectories, viewed as
functions of $t\in\left[  0,T\right]  ,$ is a family of uniformly bounded
functions (due to the compactness of the set of measures with the properties
$i)$ -- $iii)$), which also are equicontinuous. Indeed, every version of the
vector field, along which any of the measures $\mu_{n}\left(  t\right)  $ has
to evolve, is continuous and bounded in norm by the constant $\gamma_{A}=10.$
Therefore it is compact, and so has a limit point, the function $\mu\left(
t\right)  ,$ with $\mu\left(  0\right)  =\delta_{x}.$ But $\mu\left(
t\right)  $ has to be a trajectory of NLDS, and since there is just one such
trajectory starting from $\delta_{x},$ we conclude that $\mu\left(  t\right)
=\delta_{x\left(  t\right)  }.$

Therefore $\mu_{n}\left(  t\right)  \rightarrow\delta_{x\left(  t\right)  }$
for every $t\in\left[  0,T\right]  ,$ and this convergence is uniform in $t.$
The existence of the function $\bar{\varepsilon}\left(  T,\varepsilon\right)
$ follows from the compactness of the balls $\left\{  \mu:\rho_{KROV}\left(
\mu,\delta_{x}\right)  <a\right\}  $ of measures satisfying the moment condition.

\textbf{Proof of the Proposition \ref{P1}. }

\textbf{1. }To begin with, we prove that once $\varepsilon$ is small enough,
there exists the time moment $T_{1},$ at which for any initial $\mu\left(
0\right)  ,$ satisfying conditions of our Proposition,
\begin{equation}
\left\langle \exp\left\{  \alpha L\left(  \bar{x}\right)  \right\}
\right\rangle _{\mu\left(  T_{1}\right)  }<\varepsilon. \label{041}%
\end{equation}
Indeed, suppose first that the value of the exponential moment of the initial
measure $\tilde{\mu}\left(  0\right)  $ is fixed, $\left\langle \exp\left\{
\alpha L\left(  \bar{x}\right)  \right\}  \right\rangle _{\tilde{\mu}\left(
0\right)  }=M,$ and also that
\begin{equation}
\rho_{KROV}\left(  \tilde{\mu}\left(  t\right)  ,\delta_{z\left(  t\right)
}\right)  <\tilde{\varepsilon}\text{ for all }t\in\left[  0,T_{\mathcal{C}%
}\right]  ,\text{ with }\tilde{\varepsilon}\text{ small enough,} \label{040}%
\end{equation}
where $T_{\mathcal{C}}$ is the time it takes to go once around the cycle
$\mathcal{C}.$ We claim that at the moment $T_{\mathcal{C}}$ we have
$\left\langle \exp\left\{  \alpha L\left(  \bar{x}\right)  \right\}
\right\rangle _{\tilde{\mu}\left(  T_{\mathcal{C}}\right)  }<cM,$ where $c<1$
is some constant, which depends only on the parameters of our model. In
particular, $c$ is independent of $\tilde{\mu}\left(  0\right)  .$

To see that we note first that for any point $\bar{x}_{0}$ with $L\left(
\bar{x}_{0}\right)  >10$ (see definition $\left(  \ref{037}\right)  $) and
under \textit{Cyclic Dynamics} $\ \Delta_{\infty}=\Delta_{\infty}\left(
\bar{Y}\left(  \cdot\right)  \right)  ,$ for $\bar{Y}\left(  \cdot\right)
\in\mathcal{Y}\left(  \delta_{z}\right)  ,$ $z\in\mathcal{C},$ we have after
the time shift by $T_{\mathcal{C}}$ that $L\left(  \bar{x}_{T_{\mathcal{C}}%
}\right)  <L\left(  \bar{x}_{0}\right)  -C\left(  \gamma_{\ast}\right)  ,$
where the constant $C\left(  \gamma_{\ast}\right)  >0$ depends only on the
service rates $\gamma_{\ast}.$ This follows directly from the Lemma \ref{l3}
of the Section \ref{S21}. From the same Lemma and due to the condition
$\left(  \ref{040}\right)  $ we have that under any dynamics $\Delta_{\infty
}\left(  \bar{Y}\left(  \cdot\right)  \right)  $ with $\bar{Y}\left(
\cdot\right)  \in\mathcal{Y}\left(  \tilde{\mu}\left(  0\right)  \right)  $ we
have for the same point that $L\left(  \bar{x}_{T_{\mathcal{C}}}\right)
<L\left(  \bar{x}_{0}\right)  -\frac{1}{2}C\left(  \gamma_{\ast}\right)  .$ So
we have proven our claim, with $c=\exp\left\{  -\frac{1}{2}C\left(
\gamma_{\ast}\right)  \right\}  .$

Let now $k$ be the smallest integer, such that $3Ac^{k}<\varepsilon.$ We want
to repeat $k$ times the procedure of the previous paragraph. Its duration is,
evidently, $kT_{\mathcal{C}}$. Suppose that the parameter $\varepsilon$ is so
small that the function $\bar{\varepsilon}\left(  kT_{\mathcal{C}}%
,\varepsilon\right)  ,$ defined in the Proposition \ref{P2} satisfies
$\bar{\varepsilon}\left(  kT_{\mathcal{C}},\varepsilon\right)  \,<\tilde
{\varepsilon}$ (see $\left(  \ref{040}\right)  $). Then the desired repetition
is possible, and so $\left(  \ref{041}\right)  $ indeed holds, with
$T_{1}=kT_{\mathcal{C}}.$

\textbf{2. }Summarizing, at the moment $T_{1}$ we have:

$i)$
\begin{equation}
\left\langle \exp\left\{  \alpha L\left(  \bar{x}\right)  \right\}
\right\rangle _{\mu\left(  T_{1}\right)  }<\varepsilon, \label{050}%
\end{equation}

$ii)$ $\rho_{KROV}\left(  \mu\left(  T_{1}\right)  ,\delta_{z}\right)
<\bar{\varepsilon}\left(  T_{1},\varepsilon\right)  $ for some $z\in
\mathcal{C}.$

Let us use now the compact $\mathcal{K}.$ From $\left(  \ref{050}\right)  $
immediately follows that $\mu\left(  T_{1}\right)  \left[  \mathcal{K}\right]
>1-\varepsilon.$ Hence, due to the Lemma \ref{Pr4} of Section \ref{S23}, there
is the time moment $T_{2}=T_{2}\left(  \mathcal{K},\varepsilon\right)  ,$
after which the restricted measure $\mu\left(  T_{1}\right)
\Bigm|_{\mathcal{K}}$ would evolve under the dynamics $\Delta_{\infty}\left(
\bar{Y}\left(  \cdot\right)  \right)  $ with $\bar{Y}\left(  \cdot\right)
\in\mathcal{Y}\left(  \mu\left(  0\right)  \right)  $ to a $\delta$-atom, of
mass at least $1-\varepsilon$. According to Lemma \ref{P2} -- or rather its
$M=\infty$ version -- for some $z\in\mathcal{C}$ the distance $\rho
_{KROV}\left(  \mu\left(  T_{1}+T_{2}\right)  ,\delta_{z}\right)  $ at that
moment would be at most $\bar{\varepsilon}\left(  T_{2},\bar{\varepsilon
}\left(  T_{1},\varepsilon\right)  \right)  ,$ which is small. If it would
have been the case that this atom has unit mass, than after time
$T_{\mathcal{C}}$ it would already be on the cycle $\mathcal{C}.$ Since,
however, we know only that its mass is above $1-\varepsilon,$ we can claim
that at the time moment $T_{1}+T_{2}+T_{\mathcal{C}}$ its distance from
$\mathcal{C}$ is at most $K_{1}\varepsilon,$ where $K_{1}$ is some universal
constant. Therefore for the total measure $\mu\left(  T_{1}+T_{2}%
+T_{\mathcal{C}}\right)  $ we can claim that

$ii^{\prime})$ $\rho_{KROV}\left(  \mu\left(  T_{1}+T_{2}+T_{\mathcal{C}%
}\right)  ,\delta_{z}\right)  <K_{2}\varepsilon$ for some $z\in\mathcal{C}$
and some universal $K_{2},$ while

$i^{\prime})$ $\left\langle \exp\left\{  \alpha L\left(  \bar{x}\right)
\right\}  \right\rangle _{\mu\left(  T_{1}+T_{2}+T_{\mathcal{C}}\right)
}<c\varepsilon,$ with $c<1$ the same as in the first step of the proof.

So what happened is that the estimator $\varepsilon$ of the exponential
moment, appearing in $\left(  \ref{050}\right)  ,$ multiplied by $K_{2},$ is
the estimator of the $KROV$ distance on the next step, while the estimate of
the exponential moment is improved by a constant $c.$ This argument can be
iterated; on the next step we will have

$i^{\prime\prime})$ $\left\langle \exp\left\{  \alpha L\left(  \bar{x}\right)
\right\}  \right\rangle _{\mu\left(  T_{1}+2\left(  T_{2}+T_{\mathcal{C}%
}\right)  \right)  }<c^{2}\varepsilon,$

$ii^{\prime\prime})$ $\rho_{KROV}\left(  \mu\left(  T_{1}+2\left(
T_{2}+T_{\mathcal{C}}\right)  \right)  ,\delta_{z}\right)  <K_{2}c\varepsilon$
for some $z\in\mathcal{C},$

\noindent and so on, which completes the proof.
\end{proof}

\section{Euler scaling limit of the NLMP}

In this Section we will show that in the limit of high load our Non-Linear
Markov process tends to the fluid model, introduced above.

\begin{theorem}
\label{T11} Suppose that for every $N$ we are given the initial state $\nu
^{N}$ of the (scaled -- see $\left(  \ref{016}\right)  $) NLMP $\nabla
_{\infty}^{N},$ and the sequence $\nu^{N}$ converges to the measure $\mu$ in
the KROV\ metric, i.e. $\rho_{KROV}\left(  \nu^{N},\mu\right)  \rightarrow0.$
Consider all the KROV-limit points $\mu\left(  t\right)  $ of the set of
trajectories $\left\{  \tilde{\nu}^{N}\left(  t\right)  =\nu^{N}\left(
Nt\right)  ,t\in\left[  0,T\right]  \right\}  ,$ $N=1,2,...$ . Every such
limit point, called Euler fluid limit of the NLMP, is necessarily a trajectory
of the fluid model $\Delta_{\infty}.$

(Any trajectory of $\Delta_{\infty},$ obtained via this limit, will be called
a fluid solution.)
\end{theorem}

\textbf{Note. }As was mentioned above, the system $\Delta_{\infty}$ does not
possess the uniqueness property. The uniqueness property for the subclass of
the trajectories of $\Delta_{\infty}$ -- the fluid solution trajectories --
does not hold as well. That means that the above set of all limit trajectories
$\mu\left(  t\right)  $ might contain more than one element, and so the
trajectory $\mu_{\nu^{N}}\left(  t\right)  ,$ which satisfies $\mu_{\nu^{N}%
}\left(  0\right)  =\mu,$ does depend on the sequence $\nu^{N}\rightarrow\mu$.

\textbf{Proof.} \textbf{1. }We start by considering the open system. Then we
can consider every node separately. Let us take the node $\bar{A},$ say. In
the open system case its evolution is defined by prescribing the initial
state, $\nu_{\bar{A}}^{N}$ -- the distribution of the quantity $q_{\bar{A}%
}\left(  0\right)  \in\mathbb{R}^{2},$ which is the initial queue, scaled by
the factor $\frac{1}{N}\ $-- together with the rate function $\lambda_{\bar
{A}}^{N}\left(  t\right)  \equiv\left\{  \lambda_{A}^{N}\left(  t\right)
,\lambda_{BA}^{N}\left(  t\right)  \right\}  \in\mathbb{R}^{2},$ $t\geq0,$
which defines the Poisson flows of incoming clients. For our applications it
is enough to consider the case when all our rate functions $\lambda$-s are
uniformly bounded:%
\begin{equation}
\mathbf{\ }\lambda_{\ast}^{N}\left(  \ast\right)  \leq CN, \label{e01}%
\end{equation}
\textbf{ }since this is definitely the case for our closed system. The service
times of the clients are exponential, with respective rates $N\gamma_{A},$
$N\gamma_{BA}$ (this scaling is due to the Euler limit we are going to study).
As above, the $BA$ clients have priority. That means that if, while a user of
class $A$ is being served, a class $BA$ user arrives, the service of $A$ user
is interrupted until the moment when there will be no $BA$ users in queue (the
preemptive priority service discipline).

For the future use we will introduce the functions%
\[
\Lambda_{i}^{N}(t)=\frac{1}{N}\int_{0}^{t}\lambda_{i}^{N}(s)ds,\ i=A,BA,
\]
which together form a 2D vector $\Lambda^{N}\left(  t\right)  .$ We will also
denote by $Nq_{\bar{A}}^{N}(t)\in\mathbb{R}^{2}$ the pair of queues at the
moment $t.$ By $NQ_{i}^{N}(t),\ i=A,BA$ we denote the number of clients which
have arrived to the server $\bar{A}$ during the time interval $\left[
0,t\right]  .$ We can as well assume that to every client the service time is
assigned at the moment of its arrival. The sum of these required service times
for clients arrived during the time interval $\left[  0,t\right]  $ will be
denoted by $NW_{i}^{N}(t)$; this function has a stair-like graph. By
$Nw_{i}^{N}(t)$ we denote the remaining required times for clients queuing
\textit{or being served }at the moment $t;$ the graph of these functions are saw-like.

We also consider the fluid model at this node. So $q_{\bar{A}}\left(
t\right)  $ will denote the evolution of the initial measure $q_{\bar{A}%
}\left(  0\right)  $ on $\mathbb{R}^{2}$ under fluid dynamics governed by the
inflow with rate $\lambda_{\bar{A}}\left(  t\right)  .$ By $Q_{i}%
(t)=\Lambda_{i}(t)$ we denote the amounts of fluid arriving to our node during
the time interval $\left[  0,t\right]  $. Again, the fluid $BA$ has priority
over the fluid $A,$ that is, it goes out first at rate $\gamma_{1}$ whenever present.

Our goal is to prove convergence to the fluid limit. Let us assume that the
scaled inflows and initial states of the node converge to those of the fluid
model as $N\rightarrow\infty$, that is
\begin{equation}
\lim_{N\rightarrow\infty}\sup_{t\in\lbrack0,T]}\Vert\Lambda^{N}(t)-\Lambda
(t)\Vert=0 \label{e4}%
\end{equation}
and
\begin{equation}
\lim_{N\rightarrow\infty}\mathbb{\rho}_{KROV}\left(  q^{N}(0),q(0)\right)  =0.
\label{e5}%
\end{equation}

In what follows, we will need the following three bounds. The first one is the
statement that the process $W^{N}(t)$ is very close to the function
$\gamma^{-1}\Lambda^{N}(t)$ -- namely,
\begin{equation}
\lim_{N\rightarrow\infty}\mathbb{E}\left(  \sup_{t\in\lbrack0,T]}\Vert
W^{N}(t)-\gamma^{-1}\Lambda^{N}(t)\Vert\right)  {}=0. \label{E14}%
\end{equation}
A simpler claim concerns the sum $Nw^{N}(0)$ of the service times of all the
users present in the queue\textbf{ }$Nq^{N}(0)$ at the initial moment $t=0.$
Namely, for the conditional distribution of $w^{N}(0)$ under the condition
$q^{N}(0)=q$ we have\textbf{ }
\begin{equation}
\mathbb{E}\left\Vert \left(  w^{N}(0)\Bigm|q^{N}(0)=q\right)  -\gamma
^{-1}q\right\Vert {}\leq\psi(N)\Vert q\Vert\label{E15}%
\end{equation}
for some $\psi(N)\rightarrow0$ as $N\rightarrow\infty.$ Moreover,%
\begin{equation}
\sup_{t\in\lbrack0,T]}\mathbb{\rho}_{KROV}\left(  w^{N}(t),\gamma^{-1}%
q^{N}(t)\right)  \leq\psi(N)\sup_{t\in\lbrack0,T]}\mathbb{\rho}_{KROV}\left(
q^{N}(t),\mathbf{0}\right)  . \label{e30}%
\end{equation}

{\small Of course, for every fixed }$t$ {\small the convergence in }$\left(
\ref{E14}\right)  $ {\small follows from the Central Limit Theorem. The
problem is that we need the convergence at all moments }$t.$ {\small We will
obtain }$\left(  \ref{E14}\right)  $ {\small by constructing the finite-point
event, which contain the one we are interested. Indeed, consider the event
}$E\left(  t_{0},c\right)  ,$ $c>0,$ {\small which consists of all
trajectories such that }%
\begin{equation}
W_{A}^{N}\left(  t_{0}\right)  >\gamma_{A}^{-1}\left(  \Lambda_{A}^{N}\left(
t_{0}\right)  +c\right)  . \label{e02}%
\end{equation}
{\small Note that the function }$\Lambda_{A}^{N}\left(  t\right)  $
{\small has its derivative }$\leq C$ {\small (see }$\left(  \ref{e01}\right)
${\small ), so if the event }$E\left(  t_{0},c\right)  $ {\small happens, then
all the events }$E\left(  t_{0}+\tau,c-C\tau\right)  $ {\small happen as well,
since the function }$W_{A}^{N}\left(  t\right)  $ {\small is non-decreasing on
every trajectory. Therefore we can replace the infinite union by the finite
one:}%
\[%
%TCIMACRO{\dbigcup \limits_{0\leq t_{0}\leq T}}%
%BeginExpansion
{\displaystyle\bigcup\limits_{0\leq t_{0}\leq T}}
%EndExpansion
E\left(  t_{0},c\right)  \subset%
%TCIMACRO{\dbigcup \limits_{k=0}^{2CT/c}}%
%BeginExpansion
{\displaystyle\bigcup\limits_{k=0}^{2CT/c}}
%EndExpansion
E\left(  k\frac{c}{2C},\frac{c}{2}\right)  ,
\]
{\small and use the fact that for any fixed }$t\leq T$ {\small and }$c$
{\small the probability of the event }$E\left(  t,\frac{c}{2}\right)  $
{\small is exponentially small in }$N.$ {\small The remaining cases
(corresponding to the second coordinate, }$AB,$ {\small and to the lower
estimate in }$\left(  \ref{e02}\right)  ${\small ) are immediate. }

{\small The proof of }$\left(  \ref{e30}\right)  $ {\small proceeds in a
similar way. First of all, there is a natural coupling between the processes
}$w^{N}(t)$ {\small and }$q^{N}(t),$ {\small corresponding to the fact that we
can assume that the service time of every client is known at its arrival
moment. Consider the event }$E^{\prime}\left(  t_{0},c\right)  ,$
{\small consisting of the trajectories where}%
\[
w_{A}^{N}(t_{0})\geq\gamma_{A}^{-1}\left(  q_{A}^{N}(t_{0})+c\right)  .
\]
{\small We would like to use the argument similar to the above, saying that if
the bad event }$E^{\prime}\left(  t_{0},c\right)  $ {\small happens at }%
$t_{0}${\small , then on a whole segment around }$t_{0}$ {\small something
unlikely has to happen as well. Partially it can be done, since for every
trajectory we have }$w_{A}^{N}(t_{0}+\tau)\geq w_{A}^{N}(t_{0})-\tau,$
{\small and so we have for every trajectory in }$E^{\prime}\left(
t_{0},c\right)  $ {\small that }$w_{A}^{N}(t_{0}+\tau)\geq\gamma_{A}%
^{-1}\left(  q_{A}^{N}(t_{0})+c-\gamma_{A}\tau\right)  .$ {\small If we can
claim that }$q_{A}^{N}(t_{0})+c-\gamma_{A}\tau>q_{A}^{N}(t_{0}+\tau)+c/2$
{\small for all }$\tau$ {\small small enough, then we would be done. However,
the outcome that }$q_{A}^{N}(t_{0}+\tau)>q_{A}^{N}(t_{0})+c/2-\gamma_{A}\tau$
{\small is not excluded, even if }$\tau$ {\small is very small. Yet, the
probability of the increase of the queue by }$c/2-\gamma_{A}\tau$
{\small during the time }$\tau$ {\small is exponentially small in }$\tau$
{\small as }$\tau\rightarrow0.$ {\small Therefore, the event }$\cup_{0\leq
t_{0}\leq T}E^{\prime}\left(  t_{0},c\right)  $ {\small is contained in the
union }$\left[  \cup_{k=0}^{2C^{\prime}T/c}E^{\prime}\left(  k\frac
{c}{2C^{\prime}},\frac{c}{2}\right)  \right]  \cup\left[  \cup_{k=0}%
^{2C^{\prime}T/c}E^{\prime\prime}\left(  k\frac{c}{2C^{\prime}},\frac{c}%
{4}\right)  \right]  $ {\small for some suitably chosen }$C^{\prime},$
{\small where }$E^{\prime\prime}\left(  k\frac{c}{2C^{\prime}},\frac{c}%
{4}\right)  $ {\small is the event that on the segment }$\left[  k\frac
{c}{2C^{\prime}},\left(  k+1\right)  \frac{c}{2C^{\prime}}\right]  $
{\small the increment }$q_{A}^{N}(\left(  k+1\right)  \frac{c}{2C^{\prime}%
})-q_{A}^{N}(k\frac{c}{2C^{\prime}})>c/4,$ {\small and we are done.}

From (\ref{e4}) and (\ref{E14}), we get
\begin{equation}
\lim_{N\rightarrow\infty}{}\mathbb{E}\left(  \sup_{t\in\lbrack0,T]}\left\Vert
W^{N}(t)-\gamma^{-1}\Lambda(t)\right\Vert \right)  =0. \label{e2}%
\end{equation}
Note that
\begin{align*}
\mathbb{\rho}_{KROV}\left(  w^{N}(0),\gamma^{-1}q(0)\right)   &  \leq
{}\mathbb{E}\left(  \mathbb{E}\left\Vert \left(  w^{N}(0)\Bigm|q^{N}%
(0)=q\right)  -\gamma^{-1}q\right\Vert \right) \\
&  +\mathbb{\rho}_{KROV}\left(  \gamma^{-1}q^{N}(0),\gamma^{-1}q(0)\right)
{}.
\end{align*}
(Here we treat $w^{N}(0)$ in the l.h.s. as the probability distribution.) So
we get from (\ref{e5}) and (\ref{E15}) that
\begin{equation}
\mathbb{\rho}_{KROV}\left(  w^{N}(0),\gamma^{-1}q(0)\right)  \leq
\varphi(N)\mathbb{\rho}_{KROV}\left(  q(0),\delta_{\mathbf{0}}\right)  ,
\label{e3}%
\end{equation}
where $\varphi(N)\rightarrow0$ as $N\rightarrow\infty,$ and $\mathbf{0\in
}\mathbb{R}^{2}$ is the origin.

Next, we need the following estimate (see for instance \cite{LT}):

\begin{lemma}
\label{L2} Let $\Lambda(t)$ and $\Lambda^{\prime}(t)$ be two inflows to the
fluid priority node with initial (non-random) fluid levels $q(0)$ and
$q^{\prime}(0)$. Then
\[
\sup_{t\in\lbrack0,T]}\Vert q(t)-q^{\prime}(t)\Vert\leq L(\Vert q(0)-q^{\prime
}(0)\Vert+\sup_{t\in\lbrack0,T]}\Vert\Lambda(t)-\Lambda^{\prime}(t)\Vert),
\]
where $L=L(\gamma_{1},\gamma_{2})$.
\end{lemma}

\begin{proof}
Let us consider a fluid single-class node with variable capacity. Namely, let
$q_{1}(0)$ be the (scalar) initial fluid level, let $\Lambda_{1}(t)$ be the
inflow, and introduce $S_{1}(t)$ to be the server capacity, which is the
amount of work the server can do during the time interval $\left[  0,t\right]
$. (For example, in our situation $S_{1}(t)=\gamma_{1}t,$ but we will consider
more general case, with $S_{1}(t)$ not necessarily linear.) Introduce the
\textit{virtual level}
\begin{equation}
V(t)=q_{1}(0)+\Lambda_{1}(t)-S_{1}(t) \label{E7}%
\end{equation}
and the \textit{unused service capacity}%
\begin{equation}
U(t)=\max\{0,-\inf_{s\in\lbrack0,t]}V(s)\}. \label{E8}%
\end{equation}
Then
\begin{equation}
q_{1}(t)=V(t)+U(t),\quad t\geq0. \label{E9}%
\end{equation}
Let us introduce the $\sup$ norm on the space of functions. Then the
functionals $\{q_{1}(0)$, $\Lambda_{1}(\cdot)$, $S_{1}(\cdot)\}\rightarrow
V(\cdot)$ and $\{V(\cdot),U(\cdot)\}\rightarrow q_{1}(\cdot),$ given by
(\ref{E7}) and (\ref{E9}), have finite norms, since they are linear. The
non-linear functional $V(\cdot)\rightarrow U(\cdot),$ given by (\ref{E8}), has
finite norm as well. Indeed, the functional $V(\cdot)\rightarrow\inf
_{\lbrack0,\cdot]}V(\cdot)$ has norm $\leq1,$ since for any pair
$x(\cdot),y(\cdot)$ of scalar functions
\[
\left\vert \inf_{s\in\lbrack0,t]}x(s)-\inf_{s\in\lbrack0,t]}y(s)\right\vert
\leq\sup_{s\in\lbrack0,t]}|x(s)-y(s)|,
\]
and so%
\[
\sup_{t\in\lbrack0,T]}\left\vert \inf_{s\in\lbrack0,t]}x(s)-\inf_{s\in
\lbrack0,t]}y(s)\right\vert \leq\sup_{t\in\lbrack0,T]}|x(t)-y(t)|;
\]
the same holds for the functional $\left\{  x\left(  t\right)  \right\}
\rightarrow\max\{0,x(t)\},$ since
\[
\sup_{t\in\lbrack0,T]}\left\vert \max\{0,x(t)\}-\max\{0,y(t)\}\right\vert
\leq\sup_{t\in\lbrack0,T]}|x(t)-y(t)|.
\]
Therefore the composed functional, taking the triplet $\{q_{1}(0)$,
$\Lambda_{1}(\cdot)$, $S_{1}(\cdot)\}$ to the pair $\{q_{1}(\cdot)$,
$U(\cdot)\},$ defined by (\ref{E7})-(\ref{E9}), has finite norm.

That proves the desired statement for the first component of $q.$ Now, to
finish the proof, we note that the capacity of the server for the users of the
second class is given by $S_{2}(t)=\gamma_{2}U(t)$, where $U(t)$ is the unused
server capacity for the high-priority class (with $S_{1}(t)=\gamma_{1}t$,
$t\geq0$). Then we repeat the argument above.
\end{proof}

Next we formulate as a separate statement the obvious remark that the
evolution of the current remaining service time variable, $w\left(  t\right)
,$ coincides with the evolution of the level of some evidently constructed
fluid system.

\begin{lemma}
\label{L7} Let users $u_{j}$ of two possible types $i=A,BA$, with service
times $h_{i}^{j}$ arrive at the initially empty server at times $t_{i}^{j}$,
$j=1,2,\dots$, and let $w_{i}(t)$ be the evolution of the remaining service
times. Consider also the fluid model with two classes of fluids, which starts
in the empty state and is governed by the fluid inflows
\[
\Lambda_{i}(t)=\gamma_{i}\sum_{j:t_{i}^{j}\leq t}h_{i}^{j}.
\]
(I.e., our fluids have \textquotedblleft viscosities\textquotedblright%
\ $\gamma_{1}^{-1}$ and $\gamma_{2}^{-1}.$) Then at every moment $t\geq0$ the
current levels of fluids at the server equal to $\gamma_{i}w_{i}(t)$, $i=1,2$.
\end{lemma}

Another auxiliary result is needed:

\begin{lemma}
\label{L1} Let $q\in\mathbb{R}^{2}$ be (random) queue to our server, and $w$
be the corresponding (random) amount of total work (=service time needed). The
service times of the users are independent, and within the $i$-th class
identically distributed with mean $\gamma_{i}^{-1}$. Then, for any
$v\in\mathbb{R}^{2}$,
\begin{equation}
{}\mathbb{\rho}_{KROV}\left(  \gamma^{-1}q,\delta_{v}\right)  \leq
{}\mathbb{\rho}_{KROV}\left(  w,\delta_{v}\right)  . \label{e1}%
\end{equation}

\end{lemma}

\begin{proof}
Since the norm $\Vert\cdot\Vert$ is convex, we have for the conditional random
variable $w\Bigm|q=\bar{q}$ that $\left\Vert \gamma^{-1}\bar{q}-v\right\Vert
\leq\mathbb{E}\Vert\left(  w\Bigm|q=\bar{q}\right)  -v\Vert.$ Averaging over
$\bar{q}$ gives $\left(  \ref{e1}\right)  .$
\end{proof}

Now, we derive the following result

\begin{proposition}
Under the assumptions made above,
\begin{equation}
\lim_{N\rightarrow\infty}\sup_{t\in\lbrack0,T]}\mathbb{\rho}_{KROV}\left(
q^{N}(t),q(t)\right)  =0. \label{e6}%
\end{equation}

\end{proposition}

\begin{proof}
First,
\[
\sup_{t\in\lbrack0,T]}{}\mathbb{\rho}_{KROV}\left(  q^{N}(t),q(t)\right)
\leq\max\{\gamma_{1},\gamma_{2}\}\sup_{t\in\lbrack0,T]}{}\mathbb{\rho}%
_{KROV}\left(  \gamma^{-1}q^{N}(t),\gamma^{-1}q(t)\right)  .
\]
Then we write a chain of inequalities. First of all we have
\begin{align*}
&  \sup_{t\in\lbrack0,T]}{}\mathbb{\rho}_{KROV}\left(  \gamma^{-1}%
q^{N}(t),\gamma^{-1}q(t)\right) \\
&  \leq\sup_{t\in\lbrack0,T]}{}\mathbb{\rho}_{KROV}\left(  w^{N}%
(t),\gamma^{-1}q^{N}(t)\right)  +\sup_{t\in\lbrack0,T]}{}\mathbb{\rho}%
_{KROV}\left(  w^{N}(t),\gamma^{-1}q(t)\right)  ,
\end{align*}
and the first summand can be bounded by $\left(  \ref{e30}\right)  .$ To
estimate the distance $\mathbb{\rho}_{KROV}\left(  w^{N}(t),\gamma
^{-1}q(t)\right)  $ we have to exhibit some joint distribution of $\gamma
^{-1}q(t)$ and $w^{N}(t).$ We take the following one: first, we choose the
coupling between $\gamma^{-1}q(0)$ and $w^{N}(0),$ using $\left(
\ref{E15}\right)  $ and $\left(  \ref{e5}\right)  ,$ getting%
\[
\mathbb{\rho}_{KROV}\left(  w^{N}(0),\gamma^{-1}q(0)\right)  \leq
2\psi(N)\mathbb{\rho}_{KROV}\left(  q(0),\mathbf{0}\right)  .
\]
Given the joint realization of the initial values $\left(  \gamma
^{-1}q(0),w^{N}(0)\right)  ,$ the evolution of the coordinate $q(t)$ is
deterministic, defined by the flows $\Lambda_{\bar{A}}(t)$ of the arriving
fluids. The evolution $w^{N}(t)$ is stochastic, governed by the Poisson
process with net rates $\Lambda_{i}^{N}(t).$ Therefore for every $t$ we have
to exhibit the coupling between the distribution of the vector $w^{N}(t)$ and
deterministic value $q(t).$ The resulting $KROV$ distance is precisely what
the Lemmas \ref{L2} and \ref{L7} allow us to control:%
\begin{align*}
&  \mathbb{\rho}_{KROV}\left(  w^{N}(t),\gamma^{-1}q(t)\Bigm|w^{N}%
(0),q(0)\right) \\
&  \equiv\mathbb{E}\left(  \left\Vert w^{N}(t)-\gamma^{-1}q(t)\right\Vert
\Bigm|w^{N}(0),q(0)\right) \\
&  \leq L\left(  \Vert w^{N}(0)-\gamma^{-1}q(0)\Vert+\mathbb{E}\left(
\sup_{s\in\lbrack0,t]}\Vert W^{N}(s)-\gamma^{-1}\Lambda(s)\Vert\right)
\right)  .
\end{align*}
It remains to apply bounds (\ref{e2}) and (\ref{e3}) and deduce (\ref{e6}).
\end{proof}

\textbf{2. }The proof of our statement for the closed system does not require
any extra arguments, since the closed system is a special case of the open
system, where all the flows satisfy the relations defining the closed system.

$\blacksquare$

\section{Main result}

In this section we finally formulate and prove our main theorem, which claims
that the NLMP started from some special initial state behaves similarly to the
fluid model in its periodic regime, at all times $t\in\left(  0,\infty\right)
$. Since we know already that the NLMP is in turn a limit of networks of size
$M,$ as $M\rightarrow\infty,$ our theorem implies that the large size $\left(
M\gg1\right)  $ Markov process $\nabla_{M}^{N}$ behaves similarly to the fluid
model for a very long time, which time diverges as $M\rightarrow\infty,$
provided the number $N$ of clients per node exceeds some value $N_{0}.$ In
particular, there are initial states for the networks $\nabla_{M}^{N},$ which
lead to a long time oscillations, before the network reaches its stationary state.

\begin{theorem}
\label{MR} Let $\varepsilon>0.$ Then there exist the values $N_{0},$
$\varepsilon^{\prime}>0,$ $\alpha>0$ and $E<\infty$ such that for all
$N>N_{0}$ the states $\nu^{N}\left(  t\right)  $ of the NLMP process
$\nabla_{\infty}^{N},$ started at the initial state $\nu^{N}\left(  0\right)
$ with the properties:%
\[
\rho_{KROV}\left(  \nu^{N}\left(  0\right)  ,\delta_{x\left(  0\right)
}\right)  <\varepsilon^{\prime},
\]
with $x\left(  0\right)  \in\mathcal{C},$
\[
\left\langle \exp\left\{  \alpha L\left(  \bar{x}\right)  \right\}
\right\rangle _{\nu^{N}\left(  0\right)  }<E,
\]

\noindent satisfies for all $t>0$%
\begin{equation}
\rho_{KROV}\left(  \nu^{N}\left(  t\right)  ,\delta_{x_{\mathbf{\nu}}}\right)
<\varepsilon,\label{002}%
\end{equation}
where $x_{\mathbf{\nu}}\in\mathcal{C}$ is some moving point, depending on the
process $\mathbf{\nu=}\left\{  \nu^{N}\left(  t\right)  ,t\geq0\right\}  $. In
particular, the process $\nu^{N}\left(  t\right)  $ has no limit as
$t\rightarrow\infty.$
\end{theorem}

In words, we are proving that if we start the NLMP with high load $N$ per
server, from the state close to some atomic measure $\delta_{z}$ with $z$
belonging to the cycle, then it never goes to a limit.

For that, we need a general Lemma, which is formulated in the Euler scaling.
First, we recall the definitions. Let $\bar{\lambda}\left(  t\right)
=\left\{  \lambda_{i}\left(  t\right)  ,i=1,...,k\right\}  $ be the rates of
Poisson inflows of the customers of $k$ types, and $\gamma_{i},$ $i=1,...,k$
be their rates of service. The discipline of service will be irrelevant here;
we need only that the server is not idle if the queue is not empty. We call
the flow $\bar{\lambda}\left(  t\right)  $ to be underloaded, with parameters
$\left(  T,\delta\right)  ,$ if for any $t$%
\[
\sum_{i=1}^{k}\frac{1}{\gamma_{i}}\int_{t}^{t+T}\lambda_{i}\left(  t\right)
dt<\left(  1-\delta\right)  T.
\]
Below we are talking about the flow with load $N.$ That means that we consider
the situation when the input rates are given by $N\bar{\lambda}\left(
t\right)  =\left\{  N\lambda_{i}\left(  t\right)  ,i=1,...,k\right\}  ,$ while
the service rates equals to $N\gamma_{i}.$

\begin{lemma}
\label{R2} Consider the Non-Homogeneous Markov Process $\mu\left(  t\right)
$, started from the initial state $\mu\left(  0\right)  ,$ and suppose that
its generating rate function $\bar{\lambda}\left(  t\right)  $ is underloaded,
with parameters $\left(  T,\delta\right)  .$ Then there exist values
$\alpha>0$ and $A<\infty,$ depending only on the pair $\left(  T,\delta
\right)  ,$ such that if the exponential moment $\left\langle \exp\left\{
\alpha L\left(  \bar{x}\right)  \right\}  \right\rangle _{\mu\left(  0\right)
}$ of the initial state is finite, then for times $t>t\left(  \mu\left(
0\right)  \right)  $ and for any load $N$
\[
\left\langle \exp\left\{  \alpha L\left(  \bar{x}\right)  \right\}
\right\rangle _{\mu\left(  t\right)  }<A.
\]
Moreover, there exists the time $\mathcal{T}=\mathcal{T}\left(  T,\delta
\right)  ,$ such that for any initial state $\mu\left(  0\right)  $,
satisfying the estimate
\[
\left\langle \exp\left\{  \alpha L\left(  \bar{x}\right)  \right\}
\right\rangle _{\mu\left(  0\right)  }<3A,
\]
we have, for any load $N,$ that
\[
\left\langle \exp\left\{  \alpha L\left(  \bar{x}\right)  \right\}
\right\rangle _{\mu\left(  \mathcal{T}\right)  }<2A.
\]

\end{lemma}

\begin{proof}
We consider Poisson inflow with a general distribution $\eta$ of the service
time, having finite exponential moment. In particular, exponential service
time fits.

The users arrive to the node according to the Poisson processes with rates
$\lambda_{i}(t)$. Their service times are i.i.d.~with distribution functions
$\eta_{i}\left(  h\right)  $. Let $E_{i}=\mathbb{E}\left(  \eta_{i}\right)  .$
We study the dynamics of the remaining service time, hence, the service
discipline is of no importance. The regime we are interested in is the
underloaded regime; that means that for some $\delta>0$, all $T$ large enough
and all $t\geq0$
\[
\int_{t}^{t+T}\left(  \sum_{i}\lambda_{i}(s)E_{i}\right)  ds\leq T\left(
1-\delta\right)  .
\]

Let $d\mu_{t}(u)$ be the current distribution of the remaining service time.
We want to study the exponential moment
\begin{equation}
Q_{\alpha}(t)=\int_{0}^{\infty}q_{\alpha}(u)d\mu_{t}(u), \label{E11}%
\end{equation}
where $q_{\alpha}(u)=e^{\alpha u}$. We will show that the moment $Q_{\alpha
}(t)$ satisfies the equation
\[
Q_{\alpha}(t)\leq e^{C_{1}-\beta t}Q_{\alpha}(0)+C_{2},\qquad t\geq0.
\]
The\textbf{ }statistics of the observable $L$ will then be easy to derive.

Note that the underload condition ensures the absolute continuity of\textbf{
}$Q_{\alpha}(t)$. \textbf{ }We will need the quantities
\[
\Phi_{i}^{\alpha}=\int_{0}^{\infty}\left(  e^{\alpha h}-1\right)  d\eta
_{i}(h).
\]
We assume that $\Phi_{i}^{\alpha}<+\infty$ for all $\alpha\leq\overline
{\alpha}$ with $\overline{\alpha}>0$.

Before studying the moments $\left(  \ref{E11}\right)  ,$ we will consider the
situation of the \textquotedblleft broken\textquotedblright\ server, when the
clients (of one type) only come, but are not served. The queue then only grows
in time, as is the workload $u.$ The corresponding exponential moment will be
denoted by $Q_{\alpha}^{\left(  1\right)  }(t)$. We have:
\begin{equation}
\dot{Q}_{\alpha}^{\left(  1\right)  }(t)=\lambda(t)\Phi^{\alpha}Q_{\alpha
}^{\left(  1\right)  }(t). \label{E1}%
\end{equation}
Indeed, the event of arrival of a user with service time $h$ at the queue with
the current workload $u$ shifts the workload to the value $u+h,$ so the value
of $q_{\alpha}$ changes from $e^{\alpha u}$ to $e^{\alpha(u+h)}=e^{\alpha
u}+e^{\alpha u}(e^{\alpha h}-1)$. In order to find $\dot{Q}_{\alpha}^{\left(
1\right)  }(t)$, we have to multiply the increment $e^{\alpha u}(e^{\alpha
h}-1)$ by the rate $\lambda(t)$ of the arrival event and to integrate it with
respect to $d\mu_{t}(u)\times d\eta(h)$, since $u$ and $h$ are independent. In
this way we arrive to (\ref{E1}).

Next, let us study the case of \textquotedblleft broken pipe\textquotedblright%
, when the inflow is zero, so the server works only on the initial supply of
clients. The evolution of the distribution $\mu_{t}$ of the workload is given
by the following simple relation:%
\[
\mu_{t+s}\left[  a,b\right]  =\left\{
\begin{array}
[c]{cc}%
\mu_{t}\left[  a+s,b+s\right]  & \text{ if }0<a<b,\\
\mu_{t}(-\infty,b+s] & \text{ if }a\leq0<b.
\end{array}
\right.
\]
In words, the atom at $u=0$ grows with time. We denote the corresponding
exponential moment by $Q_{\alpha}^{\left(  2\right)  }(t).$ The
straightforward computation shows that
\begin{equation}
\dot{Q}_{\alpha}^{\left(  2\right)  }(t)=\alpha p_{0}(t)-\alpha Q_{\alpha
}^{\left(  2\right)  }(t), \label{E2}%
\end{equation}
where $p_{0}(t)$ is the current probability of the queue to be empty. Note
that $Q_{\alpha}^{\left(  2\right)  }(t)$ is absolutely continuous and
(\ref{E2}) holds for almost all $t$.

In the general case of several inflows we put the two relations together to
get
\begin{equation}
\dot{Q}_{\alpha}(t)=\left[  \left(  \sum_{i}\lambda_{i}(t)\Phi_{i}^{\alpha
}\right)  -\alpha\right]  Q_{\alpha}(t)+\alpha p_{0}(t). \label{E3}%
\end{equation}

Let us now rewrite $\Phi_{i}^{\alpha}$. We have%
\[
\Phi_{i}^{\alpha}=\int_{0}^{\infty}\left(  e^{\alpha h}-1\right)  d\eta
_{i}(h)=
\]%
\begin{equation}
\int_{0}^{\infty}\alpha h~d\eta_{i}(h)+\int_{0}^{\infty}\left[  e^{\alpha
h}-1-\alpha h\right]  d\eta_{i}(h)\equiv\alpha E_{i}+\alpha F_{i}(\alpha),
\label{E4}%
\end{equation}
where $E_{i}$ is the mean service time and $F_{i}(\alpha)$ is continuous
function of $\alpha\in\lbrack0,\overline{\alpha}],$ which satisfies
$F_{i}(\alpha)=O\left(  \alpha\right)  $ as $\alpha\rightarrow0$. From
(\ref{E3}) and (\ref{E4}) we get for $\alpha<1$ the bound
\begin{equation}
\dot{Q}_{\alpha}(t)\leq\alpha\left[  \left(  \sum_{i}\lambda_{i}%
(t)E_{i}\right)  -1+\sum_{i}\lambda_{i}(t)F_{i}(\alpha)\right]  Q_{\alpha
}(t)+1. \label{E5}%
\end{equation}

The Euler scaling with parameter $N$ changes $\lambda(t)$ to $\lambda
^{N}(t)=N\lambda(t)$ and $\eta(h)$ to $\eta^{N}(h)=N\eta(Nh)$. Hence,
$\lambda_{i}^{N}(t)E_{i}^{N}$ does not depend on $N.$ Let us show that
$NF_{i}^{N}(\alpha)$ is small for all $N,$ once $\alpha$ is small. Indeed,
\begin{align*}
N\alpha F^{N}(\alpha)  &  =N^{2}\int_{0}^{\infty}\left[  e^{\alpha h}-1-\alpha
h\right]  d\eta(Nh)\\
&  =N^{2}\int_{0}^{\infty}\left[  e^{\frac{\alpha}{N}Nh}-1-\frac{\alpha}%
{N}Nh\right]  d\eta(Nh)\\
&  =N^{2}\frac{\alpha}{N}F(\frac{\alpha}{N})\sim\alpha^{2}%
\end{align*}

Hence, we get a uniform bound for $\alpha$ small enough and all $N\geq1$
simultaneously (and, by the limit, for the fluid model \textquotedblleft%
$N=\infty$\textquotedblright\ as well):
\begin{equation}
\dot{Q}_{\alpha}^{N}(t)\leq\alpha\left[  \sum_{i}\lambda_{i}(t)\left(
E_{i}+\varkappa\left(  \alpha\right)  \right)  -1\right]  Q_{\alpha}^{N}(t)+1,
\label{E6}%
\end{equation}
where $\varkappa\left(  \alpha\right)  \sim\alpha$.

The solution to the linear equation
\[
\dot{x}(t)=a(t)x(t)+b\left(  t\right)
\]
is given by the formula
\[
x(t)=g(0,t)x(0)+\int_{0}^{t}g(s,t)b\left(  s\right)  ds,
\]
where $g(s,t)=e^{\int_{s}^{t}a(\tau)d\tau}$. We apply it to (\ref{E6}), with
$x\left(  t\right)  =Q_{\alpha}\left(  t\right)  $, $a(t)=\alpha\left[
\sum_{i}\lambda_{i}(t)\left(  E_{i}+\varkappa\left(  \alpha\right)  \right)
-1\right]  $ and $b(t)=1$. By the underload assumption,
\[
\int_{s}^{t}a(\tau)d\tau\leq C_{1}-\beta(t-s)
\]
for some $C_{1},\beta>0$ and for all $s<t,$ once $\alpha$ is small. Then,
$\int_{0}^{t}g(s,t)ds\leq C_{2}$ for all $t\geq0$.

Hence,
\[
Q_{\alpha}(t)\leq e^{C_{1}-\beta t}Q_{\alpha}(0)+C_{2},\qquad t\geq0.
\]

In our case the distribution $\eta_{i}$ is exponential with the parameter
$\gamma_{i}.$ Let us show finally that the exponential bound on the workload
implies an exponential bound on the number of customers (may be, with another exponent).

Indeed, under the condition that we are in the state with $n_{1}$ and $n_{2}$
customers of two classes, the conditional distribution of the workload $u$ is
a measure $\mu_{n_{1}n_{2}}$ on ${}\left(  \mathbb{R}^{1}\right)  ^{+},$ with
mean value $\bar{u}=\frac{n_{1}}{\gamma_{1}}+\frac{n_{2}}{\gamma_{2}}$. By
convexity of the exponent,
\[
\int e^{\alpha u}d\mu_{n_{1}n_{2}}(u)\geq e^{\alpha\bar{u}},
\]
which provides us with the upper bound%
\[
e^{\frac{\alpha}{\gamma_{1}+\gamma_{2}}(n_{1}+n_{2})}\leq\int e^{\alpha u}%
d\mu_{n_{1}n_{2}}(u).
\]
Taking expectations with respect to $n_{1}\left(  t\right)  ,n_{2}\left(
t\right)  $ we get
\[
\mathbb{E}\left(  e^{\frac{\alpha}{\gamma_{1}+\gamma_{2}}(n_{1}(t)+n_{2}%
(t))}\right)  \leq\int e^{\alpha u}d\mu_{t}(u)=Q_{\alpha}(t),
\]
which is the desired estimate.

\textbf{Proof of the Main Theorem. }The proof proceeds by \textquotedblleft
induction\textquotedblright\ in time. We suppose inductively that at a certain
(Euler) moment $T$ the NLMP with the load $N$ is in the state $\mu\left(
T\right)  ,$ having two properties:%
\begin{equation}
\left\langle \exp\left\{  \alpha L\left(  \bar{x}\right)  \right\}
\right\rangle _{\mu\left(  T\right)  }<3A, \label{042}%
\end{equation}%
\begin{equation}
\rho_{KROV}\left(  \mu\left(  T\right)  ,\delta_{x}\right)  <\varepsilon
\label{043}%
\end{equation}
for some $x\in\mathcal{C}$. We will show that there exists the time
$T^{\prime},$ at which the same two conditions hold for the measure
$\mu\left(  T+T^{\prime}\right)  $ -- except for different point $x$ on the
cycle $\mathcal{C}.$

To see this we first consider the Non-Linear Dynamical System (NLDS)
$\Delta_{\infty}$, with initial state $\mu\left(  T\right)  .$ In other words,
$\Delta_{\infty}=\Delta_{\infty}\left(  \bar{Y}\left(  \cdot\right)  \right)
,$ for some $\bar{Y}\left(  \cdot\right)  \in\mathcal{Y}\left(  \mu\left(
T\right)  \right)  .$ We can use the Proposition \ref{P1}, which tells us that
for any $T^{\prime}$ large enough $\rho_{KROV}\left(  \Delta_{\infty
}^{T^{\prime}}\mu\left(  T\right)  ,\delta_{x\left(  T^{\prime}\right)
}\right)  <\varepsilon/3,$ Choosing one such $T^{\prime}$ (uniformly in
$\mu\left(  T\right)  ,$ satisfying $\left(  \ref{042}\right)  -\left(
\ref{043}\right)  $ !) we can claim that for the NLMP evolution we have
$\rho_{KROV}\left(  \mu\left(  T+T^{\prime}\right)  ,\delta_{x\left(
T^{\prime}\right)  }\right)  <2\varepsilon/3,$ provided only that $N$ is large
enough; indeed, we know from Theorem \ref{T11} that the NLMP converges in the
KROV metric to NLDS on any finite time interval, as $N\rightarrow\infty,$
which convergence is uniform over the set of initial measures satisfying
$\left(  \ref{042}\right)  .$ Note that we thus have reproduced the condition
$\left(  \ref{043}\right)  .$

Proposition \ref{P2} tells us that for all $t\leq T^{\prime}$ $\rho
_{KROV}\left(  \Delta_{\infty}^{t}\mu\left(  T\right)  ,\delta_{x\left(
t\right)  }\right)  <\bar{\varepsilon}\left(  T^{\prime},\varepsilon\right)
.$ Due to the same convergence statement, for the NLMP$\ $evolution we have
$\rho_{KROV}\left(  \mu\left(  T+t\right)  ,\delta_{x\left(  t\right)
}\right)  <2\bar{\varepsilon}\left(  T^{\prime},\varepsilon\right)  $ for all
$t\leq T^{\prime}.$ In words, the measure $\mu\left(  T+t\right)  $ goes very
close to the cycle trajectory. Since we want to use the Lemma \ref{R2}, we can
as well assume that $T^{\prime}>\mathcal{T},$ where $\mathcal{T}$ is the time
introduced in this Lemma. Now all its conditions are satisfied, so Lemma
\ref{R2} tells us that $\left\langle \exp\left\{  \alpha L\left(  \bar
{x}\right)  \right\}  \right\rangle _{\mu\left(  T+T^{\prime}\right)  }<2A,$
thus the condition $\left(  \ref{042}\right)  $ is reproduced as well.
\end{proof}

\section{Conclusions}

Our main result indicates that there is an important similarity between large
queuing networks and large systems of statistical mechanics. Namely, we have
shown that the load per server plays for some networks the same role as the
inverse temperature in statistical mechanics. At high load the network can
lose the property of uniqueness of the stationary state and start to behave in
the oscillatory manner. This phenomenon looks similar to the fact that some 3D
systems with continuous symmetry are not ergodic under Glauber dynamics, when
the temperature is low enough.

It is very interesting to understand how general this phenomenon is; our
expectations are that such non-ergodic behavior is a characteristic feature of
the high load regime.

In the forthcoming publications we will show that the behavior in the low load
regime is always ergodic, which corresponds to the high temperature uniqueness
of statistical mechanics.

\textbf{Acknowledgment. }\textit{We would like to thank our colleagues -- in
particular, J. Chayes, F. Kelly, M. Biskup, O. Ogievetsky, G. Olshansky, Yu.
Peres, S. Pirogov, -- for valuable discussions and remarks, concerning the
topic of this paper. A. Rybko would like to acknowledge the financial support
and hospitality of CPT, Luminy, Marseille, where part of the work was done.}

\end{document}